\documentclass{article}

\usepackage{PRIMEarxiv}

\usepackage[utf8]{inputenc} 
\usepackage[T1]{fontenc}    
\usepackage{hyperref}       
\usepackage{url}            
\usepackage{booktabs}       
\usepackage{amsfonts}       
\usepackage{nicefrac}       
\usepackage{microtype}      
\usepackage{lipsum}
\usepackage{graphicx}

\usepackage{cite}
\usepackage{amssymb,amsthm}
\usepackage{amssymb}
\usepackage{acronym}
\usepackage{epsfig}
\usepackage[acronym,smallcaps]{glossaries}
\usepackage{lineno,hyperref}
\usepackage{algorithm}
\usepackage{dsfont}

\usepackage{multirow}
\usepackage{arydshln}
\usepackage{subfig}
\usepackage{color,soul}
\usepackage{amsmath}
\usepackage[toc,page]{appendix}
\usepackage{mathtools}

\newtheorem{prop}{Proposition}
\newtheorem{lemma}{Lemma}
\newtheorem{assumption}{Assumption}
\newtheorem{remark}{Remark}
\graphicspath{ {./Figures/} }

\title{Unsupervised Detection of Sub-Territories of the Subthalamic Nucleus During DBS Surgery with Manifold Learning
}

\author{
   Ido Cohen* \\
  The Andrew and Erna Viterby Faculty of Electrical and Computer Engineering \\
  Technion \\
  Haifa, Israel\\
  \texttt{sidoc@campus.technion.ac.il} \\
   \AND
  Dan Valsky* \\
  The Edmond and Lily Safra Center for Brain Research (ELSC) \\
  The Hebrew University \\
  Jerusalem, Israel\\
  \texttt{dan@cfin.au.dk} \\
   \And
  Ronen Talmon \\
  The Andrew and Erna Viterby Faculty of Electrical and Computer Engineering \\
  Technion \\
  Haifa, Israel\\
  \texttt{ronen@ef.technion.ac.il} \\
}
\begin{document}
\maketitle

\begin{abstract}
During Deep Brain Stimulation (DBS) surgery for treating Parkinson's disease, one important task is to detect a specific brain area called the Subthalamic Nucleus (STN) and a sub-territory within the STN called the Dorsolateral Oscillatory Region (DLOR).
Accurate detection of the STN borders is crucial for adequate clinical outcomes. 
Currently, the detection is based on human experts, often guided by supervised machine learning detection algorithms. Consequently, this procedure depends on the knowledge and experience of particular experts and on the amount and quality of the labeled data used for training the machine learning algorithms. 
In this paper, in order to circumvent such dependence and the inevitable bias introduced by training data, we present a data-driven unsupervised algorithm for detecting the STN and the DLOR during DBS surgery.
Our algorithm is based on an agnostic modeling approach for general target detection tasks.
Given a set of measurements, we extract features and propose a variant of the Mahalanobis distance between these features. We show theoretically that this distance enhances the differences between measurements with different intrinsic characteristics.
Then, we incorporate the new features and distances into a specific manifold learning method, called Diffusion Maps. We show that this particular Diffusion Maps method gives rise to a representation that is consistent with the underlying factors that govern the measurements. 
Since the construction of this representation is carried out without rigid modeling assumptions, solely from the measurements, it can facilitate a wide range of detection tasks;
here, we propose a specification for STN and DLOR detection during DBS surgery. 
We present detection results on 25 sets of measurements recorded from 16 patients during surgery.
Compared to a competing supervised algorithm based on the Hidden Markov Model (HMM), our unsupervised method demonstrates similar results in the STN detection task and superior results in the DLOR detection task.
\end{abstract}

\keywords{diffusion maps \and manifold learning \and Mahalanobis distance \and subthalamic nucleus \and
deep brain stimulation \and Parkinson’s disease; }

\section{Introduction} \label{section:introduction}
Deep Brain Stimulation (DBS) is a treatment involving an implanted stimulating device that sends electrical signals to brain areas that are responsible for body movements. Once the device is implanted in an appropriate position, DBS can help to reduce the symptoms of tremor, slowness, stiffness, and walking problems caused by several neuronal diseases, such as Parkinson’s disease, dystonia, or essential tremor. 
More specifically, we focus on the DBS of the Subthalamic Nucleus (STN), which is a known and effective treatment for Parkinson’s disease \cite{limousin1998electrical,benabid1994acute}.
During the surgical procedure to implant the DBS lead, one important task is to detect the exact location of the STN borders and a sub-territory within it, called Dorsolateral Oscillatory Region (DLOR). In \cite{hariz2002complications,nickl2019rescuing,moro2002impact,witt2012factors}, it was shown that accurate detection of these regions contributes substantially to the clinical beneﬁt of STN-DBS. 

The common procedure to detect the STN borders and the DLOR consists of two steps.
First, before the surgery, a coarse approximation of the STN location is obtained based on Magnetic Resonance Imaging (MRI) and computed tomography (CT) images.
This coarse approximation facilitates the determination of a pre-planned trajectory to the target region.
Second, during the surgery, the exact detection of the target region is based on Micro Electrode Recordings (MERs) of neuronal activity along the pre-planned trajectory.
The MERs are typically intricate and one needs to extract the relevant information for the detection of the target regions. 
In general, most target detection algorithms use prior knowledge on the data, such as predefined models, hand-crafted features, and human expert labels, in order to identify specific patterns that are indicative of the target region. 
For example, existing methods rely on various features extracted from the MERs, including the total power of the signal \cite{moran2006real}, the oscillatory activity in the beta (13–30 Hz) frequency band \cite{weinberger2006beta,zaidel2010subthalamic,shamir2012microelectrode}, and the high-frequency ($>500$ Hz) neuronal ``noise'' \cite{novak2007detection,telkes2016prediction}. Then, based on such features, as well as on human expert labels, supervised classifiers are applied \cite{valsky2017s,wong2009functional,zaidel2009delimiting}.

Current detection methods suffer from the following inherent shortcomings. First, the accuracy of supervised classifiers greatly depends on the amount of labeled data, which is typically small in STN detection applications. Second, the hand-crafted features that are used in the detection do not necessarily carry sufficient information on the STN's exact location. Third, the labels are tagged by human experts, and therefore, naturally, are biased towards their specific experience and knowledge. 

In this work, to alleviate such shortcomings we develop a data-driven unsupervised method and apply it to the STN and DLOR detection tasks performed during DBS surgery.
Our method is based on an agnostic modeling approach that assumes that the measurements typically include many sources of variability, where only a few of them are informative and facilitate the detection of specific target regions of interest. 
We assume that the measurements are the output of some unknown measurement function of hidden variables, which represent the sources of variability. We further assume that the hidden variables can be divided into two classes. The first class consists of the intrinsic variables, which are characterized by slow dynamics. The second class consists of interference and noise variables, which are manifested by high measurement variances.
Our method consists of three parts. In the first part, we present useful features that can be computed from the measurements. In the second part, we design a distance function between the features that captures the difference between the (slow) intrinsic variables and is invariant to the (fast) noise variables. Finally, in the third part, by making use of Diffusion Maps \cite{coifman2006diffusion}, we construct a global representation of the measurements and show that it is consistent with the intrinsic variables. Our premise, which we support by empirical evidence, is that the discovery of the intrinsic variables is useful for modeling, in general, and for detecting regions of interest, in particular.

The remainder of this paper is organized as follows.
In Section \ref{section:mehtod}, we present the proposed method in a general context, and show both theoretically and in experiments that the proposed method is able to reveal the intrinsic representation of the measurements.
In Section \ref{section:DBS}, we propose a specification of the method for the problem of STN and DLOR border detection, yielding two purely unsupervised algorithms. We present the detection results obtained by our algorithm and compare its performance to existing supervised Hidden Markov Model (HMM) based algorithms \cite{valsky2017s}.
Finally, in Section \ref{section:discussion}, we discuss the results and outline a few potential directions for future research.

\section{Method -- Unsupervised State Variables Approximation} \label{section:mehtod}
Our method is based on a setting consisting of measurements of a stochastic dynamical system, obtained through some unknown observation function. The propagation model of the dynamical system is unknown, and the main assumption is that the system is driven by a set of intrinsic variables. In this section, we present an agnostic algorithm that builds a new representation of these intrinsic variables from the system measurements. The new embedding of the system intrinsic variables facilitates accurate target detection of different system regimes.  
We start with a description of the general setup, and then we present our method and demonstrate it using simulations and real measurements from a simple mechanical system. Finally, we show a theoretical justification for our derivations.
In Section \ref{section:DBS}, we show a utilization of the algorithm for the STN and DLOR detection tasks.

\subsection{Problem Formulation}\label{subsecion:prob_formulate}

Consider a system with $N$ different states, and let $\boldsymbol{y}_i(t) \in \mathbb{R}^s$ denote the measurements of the system at state $i$, where $i=1,\ldots,N$ denotes the index of the state and $t$ represents time. 
Suppose that the measurements have two sources of variability.
The first source is governed by latent \emph{state variables} $\boldsymbol{\theta}_{i}(t) \in \mathbb{R}^{d_{1}}$. 
These variables characterize the system state, and their variation in time is a small perturbation
of some baseline value ${\bar{\boldsymbol{\theta}}_{i}} \in \mathbb{R}^{d_{1}}$. 
More specifically, we assume that the evolution in time of the state variables can be described by the following It\^{o} process:
\begin{align}
d\boldsymbol{\theta}_{i}(t) & = -\nabla {U}_{\bar{{\theta}}_i}(\boldsymbol{\theta}_{i}(t))dt+ {I}_{d_1 \times d_1} d \boldsymbol{w}_{i,\theta}(t),
\label{eq:ito_theta}
\end{align}
where the process drift is the gradient of the quadratic potential function ${U}_{\bar{{\theta}_i}}(\boldsymbol{\theta})=\frac{1}{2}(\boldsymbol{\theta} -\bar{\boldsymbol{\theta}_{i}})^\top(\boldsymbol{\theta} -\bar{\boldsymbol{\theta}}_{i})$ centered at the baseline value $\bar{\boldsymbol{\theta}}_{i}$, $\boldsymbol{w}_{i,\theta}(t)$ is a vector of $d_{1}$ independent Brownian motions, ${I}_{d_1 \times d_1}$ is a $d_{1} \times d_{1}$ identity matrix, and $(\cdot)^\top$ represents vector or matrix transpose.

The second source of variability is considered to be noise, represented by latent variables $\boldsymbol{\eta}_{i}(t) \in \mathbb{R}^{d_{2}}$.
Suppose that the \emph{noise variables} are characterized by high variability in time compared to the variability of the state variables. Formally, the evolution in time of the noise variables can be described by the following It\^{o} process:
\begin{align}
d\boldsymbol{\eta}_{i}(t) & = -\nabla {U}{\bar{\eta}_i}(\boldsymbol{\eta_{i}}(t))dt+\frac{1}{\epsilon}{I}_{d_2 \times d_2}d \boldsymbol{w}_{i,\eta}(t),
\label{eq:ito_eta}
\end{align}
where the drift is the gradient of a quadratic potential function given by ${U}_{\bar{\eta}_i}(\boldsymbol{\eta})=\frac{1}{2}(\boldsymbol{\eta} -\bar{\boldsymbol{\eta}}_i)^\top(\boldsymbol{\eta} -\bar{\boldsymbol{\eta}}_i)$, $\bar{\boldsymbol{\eta}}_i$ is an unknown baseline constant, ${I}_{d_2 \times d_2}$ is a $d_{2} \times d_{2}$ identity matrix, $\boldsymbol{w}_{i,\eta}(t)$ is a vector of $d_{2}$ independent Brownian motions, and $0<\epsilon \ll 1$. Note that the variance of the diffusion term of the noise variables in \eqref{eq:ito_eta} is larger than the variance of the diffusion term of the state variables in \eqref{eq:ito_theta} by a factor of $1/\epsilon^2$.
We assume that $\boldsymbol{w}_{i,\theta}(t)$ and $\boldsymbol{w}_{i,\eta}(t)$ are independent. 
We remark that \eqref{eq:ito_theta} and \eqref{eq:ito_eta} establish a prototypical propagation model of multi-scale stochastic dynamical systems \cite{mattingly2010convergence,meyn2012markov,dsilva2016data}.

Suppose that the measurements are given by $\boldsymbol{y}_{i}(t)=f(\boldsymbol{\theta}_{i}(t),\boldsymbol{\eta}_{i}(t))$, where $f:\mathbb{R}^d \rightarrow \mathbb{R}^s$ is some smooth bi-Lipschitz, possibly nonlinear, function and $d=d_1+d_2$.

For notational convenience, we denote 
\[
\boldsymbol{x}_{i}(t)=\begin{bmatrix}
     \boldsymbol{\theta}_{i}(t)  \\
     \boldsymbol{\eta}_{i}(t)
\end{bmatrix},
\]
and accordingly, we recast \eqref{eq:ito_theta} and \eqref{eq:ito_eta} as the following It\^{o} process in $d$ dimensions: 
\begin{align}
d\boldsymbol{x}_{i}(t) & = -\nabla {U(x}_{i}(t))dt+{\Lambda} d \boldsymbol{w}_{i}(t) ,
\label{eq:x_eq}
\end{align}
where ${U}(\boldsymbol{x})$ is a quadratic potential function centered at 
\[
\bar{\boldsymbol{x}}_{i}= \begin{bmatrix} \bar{\boldsymbol{\theta}}_{i} \\ \bar{\boldsymbol{\eta}}_i \end{bmatrix},
\]
and 
\[
{\Lambda}= \begin{bmatrix} {I}_{d_1 \times d_1} & 0 \\ 0 & \frac{1}{\epsilon}{I}_{d_2 \times d_2} \end{bmatrix}, \ \, \boldsymbol{w}_{i}(t)= \begin{bmatrix} \boldsymbol{w}_{i,\theta}(t) \\ \boldsymbol{w}_{i,\eta}(t) \end{bmatrix}.
\]

We assume that we have access to $M$ measurements from each state sampled in time on a discrete uniform grid. Accordingly, let $\boldsymbol{y}_{i}(t_{j}) \in \mathbb{R}^s$ denote the $j$th time measurement of the system at state $i$, 
where $t_{j}=\delta t\cdot j$, $j=\{0,1,2,....,M-1\}$, and $\delta t$ is the sampling time interval.
Our goal is to decouple the two sources of variability, given the measurements $\boldsymbol{y}_{i}(t_j)$ without prior knowledge on the system variables, and to build an embedding of the system measurements that is consistent with the state variables.
Since the state variables can be viewed as a proxy of the true state of the system, the ability to extract them may facilitate the identification of particular desired system regimes and target states.
In the sequel, we will show how such an embedding sets the stage for accurate, unbiased STN and DLOR detection.

The specification of this problem formulation in the context of STN detection during DBS surgery is as follows.
We measure from $N$ depths along the pre-planned trajectory and from each depth we acquire $M$ measurements, denoted by $\boldsymbol{y}_{i}(t_j)$, where $i$ is now the index of a specific depth.
We assume that the measurements are driven by two sources of variability.
The first source of variability is represented by the state variables $\boldsymbol{\theta}_{i}(t_j)$, which are some unknown hidden variables that characterize the STN region.
The second source of variability is represented by the noise variables $\boldsymbol{\eta}_{i}(t_j)$.
We do not have direct access to the state variables depending on the region, nor to the noise variables, and we measure them through some unknown possibly nonlinear function $f$ of \[
\boldsymbol{x}_{i}(t_j)= \begin{bmatrix}
\boldsymbol{\theta}_{i}(t_j) \\
\boldsymbol{\eta}_{i}(t_j)
\end{bmatrix}.
\]
We aim to build an embedding of the measurements $\boldsymbol{y}_{i}(t_j)$, which represent the system state, and thereby, to identify in a purely unsupervised manner the STN region.

In order to accomplish this goal, we devise a pairwise distance between system states that satisfies:
\begin{align}\label{eq:goal_dist}
	d(\boldsymbol{z}_{i},\boldsymbol{z}_{l})\thickapprox \alpha ||\bar{\boldsymbol{\theta}}_{i}-\bar{\boldsymbol{\theta}}_{l}||^{2},
\end{align}
where $\boldsymbol{z}_{i}$ is some representation of the measured data $\{\boldsymbol{y}_{i}(t_j)\}_{j=1}^{M}$ at the $i$th state and $\alpha$ is some constant.

\subsection{Proposed Algorithm}

We propose an unsupervised algorithm that is able to reveal the relations between the system's intrinsic variables without any prior knowledge. In this sub-section we focus on the utilization of our method relying on the analysis presented in Section \ref{subsection:theoretical_analysis}. 
Broadly, the proposed algorithm consists of two main stages. In the first stage, we represent each measurement by a set of features that can be computed solely from measurements and devise a distance function that achieves \eqref{eq:goal_dist}. In the second stage, we apply a manifold learning method, Diffusion Maps, that constructs a global representation of the hidden state variables based on the proposed features and distance function. 

\subsubsection{Features and Distance Function}
First, for each set of measurements $\{\boldsymbol{y}_i(t_j)\}_{j=1}^M$, we define features that can be computed solely from the measurements: 
\begin{align}
    \hat{\boldsymbol{z}}_i &=\frac{1}{M}\sum_{j=1}^{M} \boldsymbol{y}_i(t_j)\label{eq:z_i}\\ 
    \hat{{C}}_i &=\frac{1}{M-1}\sum_{j=2}^{M}[\boldsymbol{\mu}_i(t_j)-\hat{\boldsymbol{\mu}}_i)][\boldsymbol{\mu}_i(t_j)-\hat{\boldsymbol{\mu}}_i)]^{\top}\label{eq:c_i}
\end{align}
where we denote the increments between consecutive measurements by $\boldsymbol{\mu}_i(t_j)=\boldsymbol{y}_i(t_j)-\boldsymbol{y}_i(t_{j-1})$ and $\hat{\boldsymbol{\mu}}_i=\frac{1}{M-1}\sum_{j=2}^{M}\boldsymbol{\mu}_i(t_j)$, so that $\hat{\boldsymbol{z}_i}$ is the empirical mean of the measurements and $\hat{{C}_i}$ is the empirical covariance of the measurement increments.

Second, we define a metric between these features, enabling us to reveal the state variables.
Particularly, we propose to use the following modified version of the (squared) Mahalanobis distance\cite{Coifman_Singer:2008,kushnir2012anisotropic}:

\begin{align}
d(\hat{\boldsymbol{z}_{i}},\hat{\boldsymbol{z}_{l}})=\frac{1}{2}(\hat{\boldsymbol{z}}_{i}-\hat{\boldsymbol{z}}_{l})^{\top}(\hat{{C}}_{i}^{-1}+\hat{{C}}_{l}^{-1})(\hat{\boldsymbol{z}}_{i}-\hat{\boldsymbol{z}}_{l}).
\label{eq:mahalanobis}
\end{align}
One of the main assumptions underlying our work is that the latent state variables $\boldsymbol{\theta}_i(t)$ are characterized by small perturbations around some informative baseline value, whereas the noise variables $\boldsymbol{\eta}_i(t)$ are characterized by high variance. In previous work, e.g., in \cite{Coifman_Singer:2008} and \cite{dsilva2016data}, it was shown that this variant of the Mahalanobis distance implicitly attenuates hidden components with high variance without prior knowledge, motivating its utilization in our setting as well. In Section \ref{subsection:theoretical_analysis}, we show theoretically that the proposed distance attenuates the influence of the noise variables and gives rise to a distance between the hidden state variables. In Section \ref{subsection:toy} and Section \ref{section:DBS}, we present empirical results that further support the usage of this distance.

\subsubsection{Diffusion Maps}

Manifold learning is a class of nonlinear geometry-oriented dimensionality reduction methods \cite{Tenenbaum2000,Roweis2000,Donoho2003,Belkin_Niyogi_2003}.
For the purpose of finding a global parametrization that embodies the relation between the system variables, we use a kernel-based manifold learning technique called Diffusion Maps \cite{Coifman_Lafon2006,coifman2006diffusion}. 
Typically in manifold learning, a high dimensional data set that is assumed to lie on a low dimensional manifold is given. This class of methods attempts to reveal the intrinsic structure of the data set (the low dimensional manifold) by preserving distances within local neighborhoods. Manifold learning methods have been successfully applied to a broad range of applications, e.g., the discovery of the latent variables of dynamical systems \cite{talmon2015manifold,yair2017reconstruction}, earth structure classification \cite{kushnir2012anisotropic}, image reconstruction \cite{zhu2018image}, signal denoising \cite{singer2009diffusion}, numerical simulation enhancement \cite{ibanez2018manifold}, fetal electrocardiogram analysis \cite{shnitzer2019recovering}, sleep stage identification \cite{wu2014assess}, and time series filtering \cite{talmon2013empirical}, to name but a few.
%

In the sequel, we will briefly review the method in the context of our work.

Suppose that we have the features of $N$ system states, i.e., ($\hat{\boldsymbol{z}_i}, \hat{{C}_i})$ for $i=1,\dots,N$, computed from the measurements.
We denote by ${W}$ the $N \times N$ pairwise affinity matrix between the features, whose $(i,l)$th element is given by:
\begin{align}\label{eq:W_gloabal}
{W}_{i,l}=\exp\left\{-{\frac{d(\hat{\boldsymbol{z}}_i,\hat{\boldsymbol{z}}_l)}{\epsilon}}\right\},
\end{align}
where the (squared) distance is defined in \eqref{eq:mahalanobis}, and $\epsilon>0$ is the kernel scale, usually set as the median of the pairwise distances.
We define a corresponding diffusion matrix ${K}$ by: 
\begin{equation}\label{eq:K}
{K}_{i,l}=\frac{{W}_{i,l}}{\boldsymbol{w}(i)}
\end{equation}
where
\begin{equation}\label{eq:w_of_K}
  \boldsymbol{w}(i)=\sum_{l=1}^{N}{W}_{i,l}
\end{equation}
We remark that several different normalizations of the affinity matrix $W$ were proposed in \cite{Coifman_Lafon2006} and in related literature. In our work, we tested different normalizations and constructed $K$ as presented since it yielded the best empirical results.

Based on the spectral decomposition of ${K}$, we build a global representation of the system states as follows.
Let $\lambda^0,...,\lambda^{N-1}$ and $\boldsymbol{\psi}^0,...,\boldsymbol{\psi}^{N-1}$ be the eigenvalues and eigenvectors of ${K}$, respectively, written in descending order, so that $\lambda_{N-1} \leq ...\leq \lambda_0=1 $.
Using the $P$ eigenvectors corresponding to the largest $P$ eigenvalues, we define the following (nonlinear) map for each state $i$ to a $P$-dimensional space:
\begin{align*}
i \mapsto \boldsymbol{\Psi}_i = (\boldsymbol{\psi}^1(i),\boldsymbol{\psi}^2(i),...,\boldsymbol{\psi}^P(i)) \in \mathbb{R}^P.
\end{align*}
Since this embedding of the data is based on an affinity that is locally invariant to the noise variables, i.e., the corresponding distance satisfies \eqref{eq:goal_dist} (see Section \ref{subsection:theoretical_analysis}), we view it as a new representation of the hidden state variables.
We conclude this section with the presentation of the proposed algorithm in Algorithm \ref{alg:proposed}.

\begin{algorithm}[t]
	\caption{The Proposed Algorithm}
	\label{alg:proposed}
	\textbf{Input}: $M$ measurements of $N$ different states, i.e., ${\boldsymbol{y}_{i}(t_{j})}_{j=1}^M \in \mathbb{R}^s, \ i={1,\ldots,N}$. \\
	\textbf{Output}: A low dimensional representation of each state $\boldsymbol{\Psi}_i \in \mathbb{R}^P$.
	
	\begin{enumerate}
		\item For each state $i$, compute the feature $\hat{\boldsymbol{z}}_i$ and the covariance matrix $\hat{{C}_i}$ according to \eqref{eq:z_i} and \eqref{eq:c_i}.
		\item Build the pairwise affinity matrix ${W}$ between all states according to \eqref{eq:W_gloabal} and \eqref{eq:mahalanobis}.
		\item  Compute the diffusion operator ${K}$ according to \eqref{eq:K} and \eqref{eq:w_of_K}
		\item Calculate the spectral decomposition of ${K}$ and obtain its eigenvalues $\left\{\lambda^l\right\}_{l=0}^{N-1}$ and right eigenvectors $\left\{\boldsymbol{\psi}^l\right\}_{l=0}^{N-1}$.
		\item Build a nonlinear mapping (embedding) of the system state:
		\begin{align*}\label{}
            (\hat{\boldsymbol{z}_i},\hat{{C}_i})\mapsto \boldsymbol{\Psi}_i=(\boldsymbol{\psi}^1(i),\boldsymbol{\psi}^2(i),...,\boldsymbol{\psi}^P(i))
		\end{align*}

	\end{enumerate}
\end{algorithm}

\subsection{Simulation Results} \label{subsecion:simulation}

\begin{figure*}[!t]
\centerline{\includegraphics[width=181mm,height=90mm]{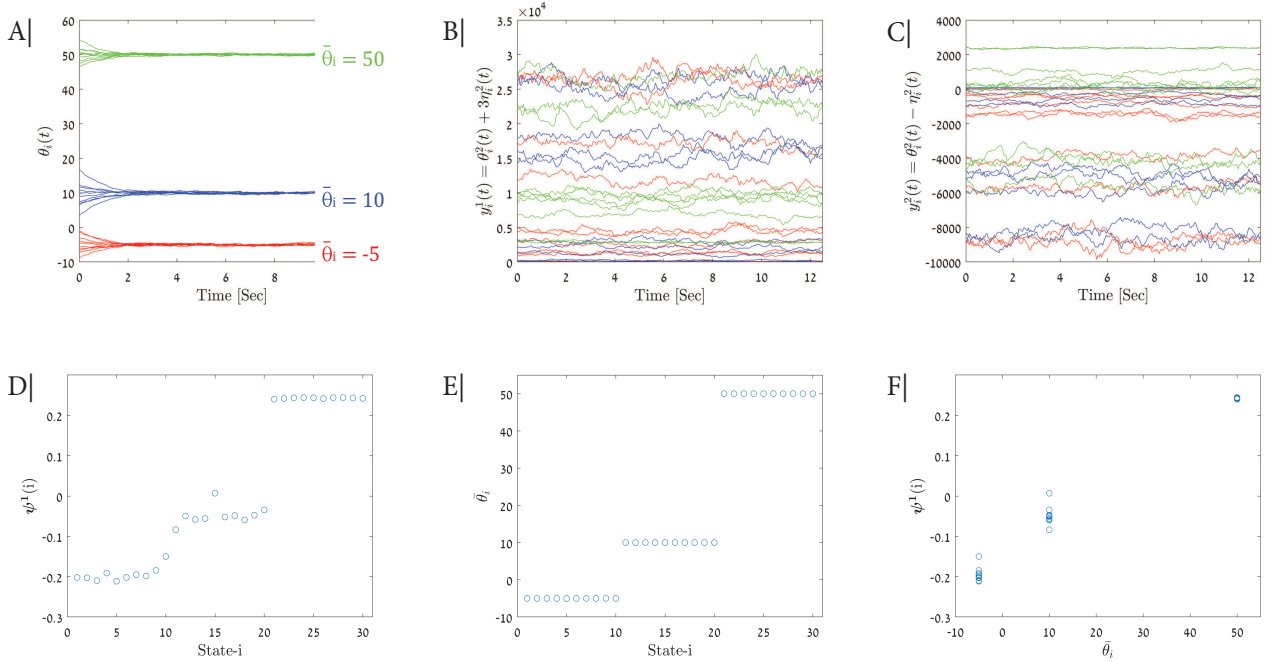}}
\caption{Illustration of our method on  simulations. (A) The system's hidden state variable $\theta_i(t)$ colored according to the baseline value $\bar{\theta}_i$.
(B)-(C) The system's measurements $ \boldsymbol{y}_i(t)=f(\theta_i(t),\eta_i(t)) \in \mathbb{R}^2$ colored according to their (hidden) baseline values $\bar{\theta}_i$. These measurements serve as the input data to our algorithm.
(D) The one dimensional embedding $\boldsymbol{\psi}^1(i)$ obtained by our suggested algorithm as a function of the state index.
(E) The true baseline values $\bar{\theta}_i$ (``the ground truth'') as a function of the state index. 
(F) A scatter plot that present the relation between the embedding obtained by our suggested algorithm and the ``ground truth''.}
\label{fig:Simulation}
\end{figure*}

To illustrate the proposed algorithm, we consider the following evolution of the state variable:
\begin{align*}
\theta_{i}(t_{j+1}) - \theta_{i}(t_{j}) & = -(\theta_{i}(t_{j})-\bar{\theta}_i)\Delta t+ \sqrt{\Delta t} w_{i,\theta},
\end{align*}
where $w_{i,\theta} \sim N(0,0.09)$, $\Delta t = 0.05$, $1 \leq j \leq 250$ and the baseline values are given by
\begin{align*}
\bar{\theta}_i =
  \begin{cases}
    -5   & \quad \text{for } 1\leq i\leq 10\\
    10  & \quad \text{for } 11\leq i\leq 20\\
    50  & \quad \text{for } 21\leq i\leq 30.
  \end{cases}
\end{align*}
A realization of all the trajectories of the state variable $\theta_i(t_j)$  is shown in Figure \ref{fig:Simulation}(A). 

In addition, we consider the following evolution of the noise variable:
\begin{align*}
\eta_{i}(t_{j+1}) - \eta_{i}(t_{j}) & = -(\eta_{i}(t_{j})-\bar{\eta}_i)\Delta t+ \frac{1}{\epsilon} \sqrt{\Delta t} w_{i,\eta},
\end{align*}
where $w_{i,\eta} \sim N(0,0.09)$, $\epsilon=0.1$, and the baseline values of the noise are uniformly sampled from $\bar{\eta}_i \in \{0,\ldots,100\}$.

Suppose that the hidden variables $(\theta_i(t_j),\eta_i(t_j))$ are observed through the following non linear function $f:\mathbb{R}^2 \rightarrow \mathbb{R}^2$:
\begin{align*}
\boldsymbol{y}_i(t_j)&=f(\theta_i(t_j),\eta_i(t_j))\\
&=(\theta_i^2(t_j)+3\eta_i^2(t_j),\theta_i^2(t_j)-\eta_i^2(t_j)).
\end{align*}
We note that this nonlinear observation function $f$ satisfies Assumption \ref{assumption:1} that is presented in Section \ref{subsection:theoretical_analysis}. 
In Figure \ref{fig:Simulation}(B)-(C), we plot the two coordinates of the measurements  $\boldsymbol{y}_i(t_j)$ in $\mathbb{R}^2$ for $i=\{1,\ldots,30\} $.
For each state $i$, we computed the features $\boldsymbol{z}_i$ and $\boldsymbol{C}_i$ based on the measurements $\boldsymbol{y}_i(t_j)$ and applied Algorithm 1.
In Figure \ref{fig:Simulation}(D)-(F), we display a comparison between the output of Algorithm 1 for $P=1$, namely, $\boldsymbol{\psi}^1(i)$, and the true (inaccessible) baseline value $\bar{\theta}_i$. 

We can observe that the computed one dimensional embedding has high correspondence with the hidden intrinsic baseline state $\bar{\theta}_i$, and can even be approximated by a linear function: $\boldsymbol{\phi}^1_i \thickapprox \alpha \bar{\theta}_i , \quad  \alpha =0.005$, thereby achieving our main goal. 

We remark that, in this specific example, our empirical results suggest that the intrinsic state of the system can be captured solely by the first leading eigenvector. However, in general, the information on the intrinsic state is often manifested in the first few leading eigenvectors. Therefore, subsequent high-dimensional clustering, such as k-means, is typically applied to the first leading eigenvectors. Importantly, such a generic clustering stage does not consider the temporal order of the samples, which is exploited in our algorithms presented in Section \ref{section:DBS}.

\subsection{Experimental Results on a Mechanical System}\label{subsection:toy} 

\begin{figure*}[!t]
\centerline{\includegraphics[width=181mm,height=90mm]{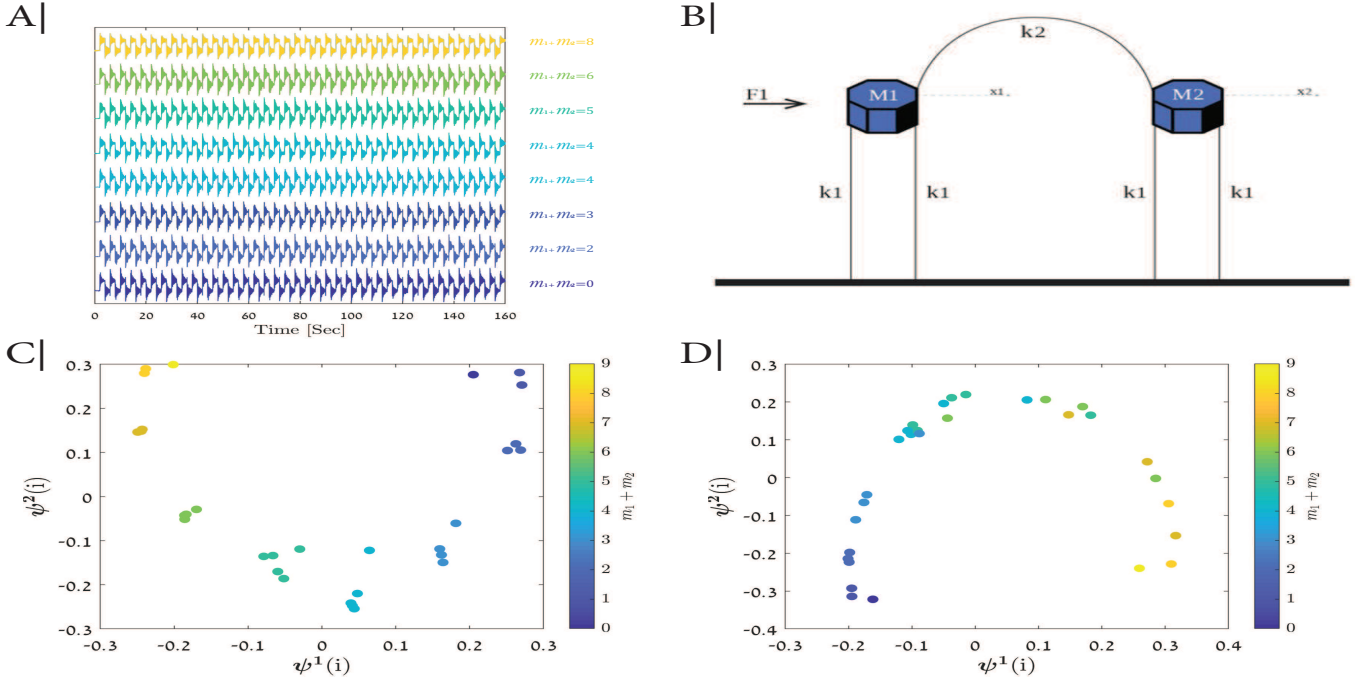}}
\caption{Illustration of our method on a mechanical system. (A) An example of the input data to the algorithm: measurements of the location of mass $m_2$ over time.
(B) A diagram of the mechanical system consisting of two coupled masses.
(C) The two-dimensional embedding obtained by Algorithm \ref{alg:proposed} (our method) colored by the sum of the masses.
(D) The two-dimensional embedding obtained by a modified Algorithm \ref{alg:proposed} colored by the sum of the masses, where the $\ell_2$ norm is used instead of the Mahalanobis distance (see text for details).}
\label{fig:Toy_Problem}
\end{figure*}

To further demonstrate the proposed method, we apply Algorithm \ref{alg:proposed} to real  measurements of a mechanical system. 
On the one hand, we show here the recovery of the main properties of the system from its observations in a data-driven manner without prior knowledge on the system. On the other hand, this particular mechanical system was chosen since it has a known definitive characterization, which can serve as a ground truth in our experiment in order to assess and validate the empirical results.

The mechanical system we consider consists of two masses, $m_1$ and $m_2$, that are coupled with a spring with constant $k_2$. Each mass is connected to the ground with two additional springs, each with constant $k_1$. 
Let $x_1$ and $x_2$ denote the positions of the two masses. An external force, denoted by $F1$, is applied to the mass $m_1$. A diagram of the mechanical system is presented in Figure \ref{fig:Toy_Problem}(B).

The experiment comprised repeated trials.
In each trial, the values of the two masses were set from a predefined grid consisting of $30$ points, where $m_1 \in \{0,\ldots,4\}$ and $m_2 \in \{0,\ldots,5\}$. 
An external force (a square wave function) was used to invoke the system using a voice-coil actuator.
With an optic-laser sensor, we measured the position of the mass $m_2$ over time.
The duration of each trial was $100$ seconds and the sampling rate was $10$ kHz. Let $g_i(t_j)$ denote the time-series of the measured signal at the $i$th trial for $t_j=1,\ldots,1,000,000$. 

In analogy to our setting, we have observations of a mechanical system from $30$ different states, where each state $i$ is specified by the values of the two masses $m_1$ and $m_2$.

Figure \ref{fig:Toy_Problem}(A) shows an example of the system measurements from different states, colored by the sum of the masses.

We follow common-practice in manifold learning and apply a pre-processing stage to the one dimensional time-series of observations.
Specifically, we computed the spectrogram of each time series using an analysis window of length $1000$ with an overlap of $500$.
Let $\boldsymbol{y}_i(t_j) \in \mathbb{R}^s$ denote the resulting spectrogram at time $t_j=1\ldots, M$ in state $i=1\ldots,N$.

\subsubsection{System Analysis}
Using Newton's law, the ODE that describes the movement of each mass is given by:
\begin{align*}
m_1 \ddot{x}_1    &= F(t) -2k_1x_1 -k_2(x_1-x_2) -c_1\dot{x}_1  \\
m_2 \ddot{x}_2    &=  -2k_1x_2 -k_2(x_1-x_2) -c_2\dot{x}_2.
\end{align*}
We omit the dumping factor of each mass, namely, $c_1$ and $c_2$, and recast the ODEs in a matrix form:
\begin{align*}
\begin{bmatrix} m_1 & 0 \\ 0 & m_2 \end{bmatrix}\begin{bmatrix} \ddot{x_1} \\ \ddot{x_2} \end{bmatrix}+\begin{bmatrix} 2k_1+k_2 & k_2 \\ k_2 & 2k_1+k_2 \end{bmatrix}\begin{bmatrix} x_1 \\ x_2 \end{bmatrix}=\begin{bmatrix} F(t) \\ 0 \end{bmatrix}
\end{align*}
The modes of the system can be found by solving an eigenvalue problem of the matrix ${K}^{-1}{M}$, where 
\begin{align*}
	{M} =\begin{bmatrix} m_1 & 0 \\ 0 & m_2 \end{bmatrix}, \, {K} = \begin{bmatrix} 2k_1+k_2 & k_2 \\ k_2 & 2k_1+k_2 \end{bmatrix}
\end{align*}
The corresponding characteristic polynomial is:
\begin{align*}
&\text{det}({K}^{-1}{M} -\lambda {I} )\\&=\text{det}\begin{bmatrix} (2k_1+k_2) \cdot m_1-\lambda & -k_2\cdot m_1 \\ -k_2\cdot m_2 & (2k_1+k_2)\cdot m_2-\lambda \end{bmatrix}\\
&=[(2k_1+k_2)\cdot m_1-\lambda ] \cdot [(2k_1+k_2) \cdot m_2-\lambda)]\\&-k_2^2 \cdot m_1 \cdot m_2\\
&=\lambda^2-\lambda \cdot (2k_1+k_2) \cdot (m_1+m_2)-k_2^2 \cdot m_1 \cdot m_2,
\end{align*}
implying that the system has two degrees of freedom (the roots of the characteristic polynomial), which are given by:
\begin{align*}
&\lambda_{1,2}=\frac{1}{2}(2k_1+k_2)\cdot(m_1+m_2)\\
&\pm \frac{1}{2}\sqrt{(2k_1+k_2)^2\cdot(m_1+m_2)^2+4\cdot k_2^2\cdot m_1\cdot m_2}\\
&=\frac{1}{2}(2k_1+k_2)\cdot(m_1+m_2) \\
&\pm \frac{1}{2}(2k_1+k_2)\cdot(m_1+m_2)  \sqrt{1+\frac{4\cdot k_2^2\cdot m_1\cdot m_2}{(2k_1+k_2)^2\cdot(m_1+m_2)^2}}
\end{align*}

Assuming that the springs remain constant during the experiment, we note that the two modes of the system are governed by the sum of the masses $m_1+m_2$, since the term consisting of their product $m_1 \cdot m_2$ is of smaller order of magnitude. This implies that the hidden state variable in each trial could be parameterized by $\bar{\theta}_i=m_1+m_2$.

\subsubsection{Results}
We apply Algorithm \ref{alg:proposed} to the input data $\{\boldsymbol{y}_i(t_j)\}_{j=1}^M$. As a baseline, we apply a similar algorithm to the same data, but instead of using the modified Mahalanobis distance we use the Euclidean distance between the features $\boldsymbol{z}_i$ (in Step 2 of Algorithm \ref{alg:proposed}, the affinity matrix ${W}$ is computed using $\|\boldsymbol{z}_i-\boldsymbol{z}_l\|_2^2$ instead of $d(\boldsymbol{z}_i,\boldsymbol{z}_l)$).
Figure \ref{fig:Toy_Problem}(C)-(D) displays a two dimensional representation of the measurements resulting from the applications of the two algorithms. 
Figure \ref{fig:Toy_Problem}(C) shows the results of Algorithm \ref{alg:proposed} and Figure \ref{fig:Toy_Problem}(D) shows the results of the baseline algorithm. Each point in the figures represents a state (trial). The points are colored by the corresponding value of $m_1+m_2$.

We observe that the two dimensional representation obtained by Algorithm \ref{alg:proposed} is organized according to the sum of the masses in each state, a result that is consistent with the analysis of the system presented above. Moreover, we observe that using the modified Mahalanobis distance rather than the Euclidean distance between the features $\boldsymbol{z}_i$ is critical for the recovery of the true hidden state of the system.

\subsection{Theoretical Analysis of the Extraction of the State Variables}\label{subsection:theoretical_analysis}

In this sub-section, we present the theoretical analysis supporting the proposed method.
We recall that our goal is to find a pairwise distance between features of the measurements that satisfies \eqref{eq:goal_dist}, thereby revealing the distance between the hidden state variables. 

For this purpose, we compute a suitable set of features for each set of measurements that carries sufficient information about the state variables. 
Assuming that the process $\boldsymbol{y}_i(t)$ is stationary and ergodic and that we have an infinitely large number of measurements from each state, the two features in \eqref{eq:z_i} and \eqref{eq:c_i} can be recast as:
\begin{align}
\boldsymbol{z}_i&=\mathbb{E}[\boldsymbol{y}_i(t)] \label{eq:z_linear}\\
{C}_i&= \mathrm{Cov} [\boldsymbol{y}_i(t+\delta t)|\boldsymbol{y}_i(t)]. \label{eq:c_overoll}
\end{align}
Namely, the expected value of the It\^{o} process measurements at a specific state, and the covariance matrix of the measurement increments.

Then, as in \eqref{eq:mahalanobis}, use a modified version of the Mahalanobis distance\cite{Coifman_Singer:2008} between the features:
\begin{align}
d(\boldsymbol{z}_{i},\boldsymbol{z}_{l})=\frac{1}{2}(\boldsymbol{z}_{i}-\boldsymbol{z}_{l})^{\top}({C}_{i}^{-1}+{C}_{l}^{-1})(\boldsymbol{z}_{i}-\boldsymbol{z}_{l}).
\label{eq:mahalanobis2}
\end{align}

In the sequel, we show that this distance indeed reveals the distance between the hidden state variables.

Our analysis is divided to two cases as follows.

\subsubsection{Direct Access}
In order to simplify the exposition, we consider the case in which the measurements $\boldsymbol{y}_i(t)$ have direct access to system variables and are equal to $\boldsymbol{x}_{i}(t)=(\boldsymbol{\theta}_{i}(t),\boldsymbol{\eta}_{i}(t)) \in \mathbb{R}^{d}$.
\begin{prop}
Given $\boldsymbol{x}_{i}(t)=(\boldsymbol{\theta}_{i}(t),\boldsymbol{\eta}_{i}(t)) \in \mathbb{R}^{d}$, the modified Mahalanobis distance in \eqref{eq:mahalanobis2} between the features $\boldsymbol{z}_i=\mathbb{E}[\boldsymbol{x}_i(t)]$ using the covariance matrices ${C}_i=\mathrm{Cov}[\boldsymbol{x}_i(t+\delta t)|\boldsymbol{x}_i(t)]$ can be written in terms of the Euclidean distance between the underlying state variables as follows:  
\begin{align}\label{eq:prop_1}
d(\boldsymbol{z}_{i},\boldsymbol{z}_{l})=\frac{1}{\delta t}\left[\|\bar{\boldsymbol{\theta}}_{i}-\bar{\boldsymbol{\theta}}_{l}\|^{2}+O(\epsilon)\right].
\end{align} 
\label{prop:1}
\end{prop}
See Appendix \ref{app1} for the proof.

Proposition \ref{prop:1} implies that by assuming direct access to the state and noise variables, the modified Mahalanobis distance with the proposed features satisfies \eqref{eq:goal_dist}. Note that in this case, although we have access to the system variables $\boldsymbol{x}_i(t)$, we do not know which of them is a state variable $\boldsymbol{\theta}_i$ and which is considered noise. Our features and distance function enable us to separate the two kinds of variables and to obtain a distance between the state variables in an implicit manner without prior knowledge.

\subsubsection{Non-Linear Measurements}
We now consider the case where the observations are a function of the system variables, i.e., $\boldsymbol{y}_{i}(t)=f(\boldsymbol{x}_{i}(t))$, where $f:\mathbb{R}^d \rightarrow \mathbb{R}^s$ is some smooth bi-Lipschitz function. Considering this case requires an additional assumption on the function $f$ and a small modification of the features $\boldsymbol{z}_i$, which are described next. 

\begin{assumption}
For any two realizations  $\boldsymbol{x}_i$ and $\boldsymbol{x}_l$ of state and noise variables, we have:
\begin{align}
&\frac{({f}_{p_1p_2}^k(\boldsymbol{x}_i)-{{f}}_{p_1p_2}^k(\boldsymbol{x}_l))^2}{{f}_{p_1}^k(\boldsymbol{x}_i){f}_{p_2}^k(\boldsymbol{x}_i)+{f}^k_{p_1}(\boldsymbol{x}_l){f}^k_{p_2}(\boldsymbol{x}_l)} \ll 1, 
\end{align}
for $1\leq k\leq s$ and $1 \leq p_{1},p_{2} \leq d$, where the subscripts correspond to partial derivatives, i.e., ${f}_p^k=\frac{\partial f^k}{\partial x^p}$ and ${f}_{p_1p_2}^k=\frac{\partial^2 f^k}{\partial x^{p_1}\partial x^{p_2}}$, and we recall that superscripts correspond to specific elements in a vector, i.e.,
${f}^k:\mathbb{R}^d \rightarrow \mathbb{R}$ is the $k$-th element of $f$.
\label{assumption:1}
\end{assumption}

Assumption 1 could be viewed as an additional smoothness property of the observation function. This assumption holds for functions that have small local changes in their gradient, and it includes any second-order polynomial function.

\begin{prop}\label{prop:2}
Given observations $\boldsymbol{y}_{i}(t)=f(\boldsymbol{x}_{i}(t))$, where $f:\mathbb{R}^d \rightarrow \mathbb{R}^s$ is a smooth function satisfying Assumption \ref{assumption:1}, the modified Mahalanobis distance with the features 
\begin{align}
\label{eq:z_i_nonlinear}
\boldsymbol{z}_{i}&=\mathbb{E}\left[\lim_{\delta t\to 0} \frac{\boldsymbol{y}_i(\delta t + t)-\boldsymbol{y}_i(t)}{\delta t} +\boldsymbol{y}_i(t)\right]\\
{C}_i&=\textrm{Cov}[\boldsymbol{y}_i(t+\delta t)|\boldsymbol{y}_i(t)] \nonumber
\end{align}
can be expressed as:  
\begin{align}
d(\boldsymbol{z}_i,\boldsymbol{z}_l)&=\|\bar{\boldsymbol{\theta}_i}-\bar{\boldsymbol{\theta}_l}\|^2  +O(\|\boldsymbol{y}_i-\boldsymbol{y}_l\|^4) + O(\epsilon).
\end{align}
\label{prop:2}
\end{prop}
See Appendix \ref{app2} for the proof.

Both Proposition \ref{prop:1} and Proposition \ref{prop:2} imply that, either with direct access or through some unknown observation function, the modified Mahalanobis distance between the proposed features reveals the distance between the state variables and attenuates the contribution of the noise variables.

We conclude this subsection with a couple of remarks.
\begin{remark}
In the specific case where $f$ is the identity function, i.e., $f(\boldsymbol{x})=\boldsymbol{x}$, different features $\boldsymbol{z}_i$ can be computed: either \eqref{eq:z_linear} or \eqref{eq:z_i_nonlinear}). Indeed, using the additional prior information of having direct access to the state and noise variables leads to a better approximation in Proposition \ref{prop:1} compared to Proposition \ref{prop:2}.
In the sequel, we will show that in practice the two definitions of $\boldsymbol{z}_i$ coincide.
\end{remark}

\begin{remark}
Our method relies on the modified Mahalanobis distance proposed in \cite{Coifman_Singer:2008}. In \cite{dsilva2016data}, this distance was used and analyzed in the context of temporal data stemming from a multiscale dynamical system. The theoretical and algorithmic parts of the present work extend \cite{Coifman_Singer:2008} and \cite{dsilva2016data} in two main aspects. First, our work considers a different model than the models considered in \cite{Coifman_Singer:2008} and in \cite{dsilva2016data}. Here, the model consists of controlling variables, which are separated to (desired) intrinsic state variables and to (undesired) noise variables. In addition, the model includes different dynamic regimes. By relying on the analysis in \cite{dsilva2016data}, we present new theoretical results, and consequently, derive new features that are specific for the model considered here. Importantly, in the sequel, we empirically support our model, features, and theoretical results by showing experimental results on DBS.
\end{remark}

\subsubsection{Feature Estimators}
In order to estimate the features in \eqref{eq:z_linear} and \eqref{eq:c_overoll} from the measurements at hand, we use the following assumptions.
First, we assume that the number of measurements at each state $M$ is large. Second, we assume that the measured signal $\boldsymbol{y}_i(t_j)$ is stationary and ergodic with respect to $t_j$ at a fixed state $i$.

For $\boldsymbol{z}_i$, we propose the following estimator:
\begin{align}\label{eq:z_i_estimator}
    \hat{\boldsymbol{z}}_i &=\frac{1}{M}\sum_{j=1}^{M} \boldsymbol{y}_i(t_j) ,
\end{align}
where the relation between the estimator and the desired feature is given by:
\begin{align}{\label{eq:z_i_hat_relation}}
    \boldsymbol{z}_i=\hat{\boldsymbol{z}}_i + O\left(\frac{1}{\delta t M}\right).
\end{align}
Since $M$ is assumed to be large, $\hat{\boldsymbol{z}}_i$ is considered as a good approximation of $\boldsymbol{z}_i$.

The relation in (\ref{eq:z_i_hat_relation}) stems from the following derivation. Since $M$ is large, then by the Law of Large Numbers, the empirical mean converges to the expected value, and so the desired feature can be recast as:
\begin{align}
    \boldsymbol{z}_i&=\mathbb{E}\left[\frac{1}{\delta t}(\boldsymbol{y}_i(\delta t+t)-\boldsymbol{y}_i(t))+\boldsymbol{y}_i(t)\right] \nonumber\\&= \frac{1}{M}\sum_{j=1}^{M} \left[\frac{\boldsymbol{y}_{i}(t_j)-\boldsymbol{y}_{i}(t_{j-1})}{\delta t}+\boldsymbol{y}_{i}(t_{j-1})\right]\nonumber \\
    &=\frac{1}{M}\sum_{j=1}^{M} \left[\frac{\boldsymbol{y}_{i}(t_j)-\boldsymbol{y}_{i}(t_{j-1})}{\delta t}\right] +\frac{1}{M}\sum_{j=1}^{M} \boldsymbol{y}_{i}(t_{j-1}) \nonumber\\
    &=\frac{1}{M}\sum_{j=1}^{M} \left[\frac{\boldsymbol{y}_{i}(t_j)-\boldsymbol{y}_{i}(t_{j-1})}{\delta t}\right] +\hat{\boldsymbol{z}}_i.\label{eq:z_i_derive2}
\end{align}
The first term in \eqref{eq:z_i_derive2} is a telescopic sum; after cancellation, this term is equal to $O(\frac{1}{\delta t M})$, leading to \eqref{eq:z_i_hat_relation}.
We obtain that the feature estimate $\hat{\boldsymbol{z}}_i$ in \eqref{eq:z_i_estimator} is a good estimate of the proposed feature \eqref{eq:z_i_nonlinear} in the nonlinear case by \eqref{eq:z_i_hat_relation}. We note that the same feature estimate in \eqref{eq:z_i_estimator} can serve as a good estimate of the proposed feature \eqref{eq:z_linear} in the direct access case as well. Therefore, in practice, we can use the same feature estimate for both cases.


\section{Experimental Results on DBS} \label{section:DBS}

\begin{figure*}[!p]
\centerline{\includegraphics[width=181mm,height=165mm]{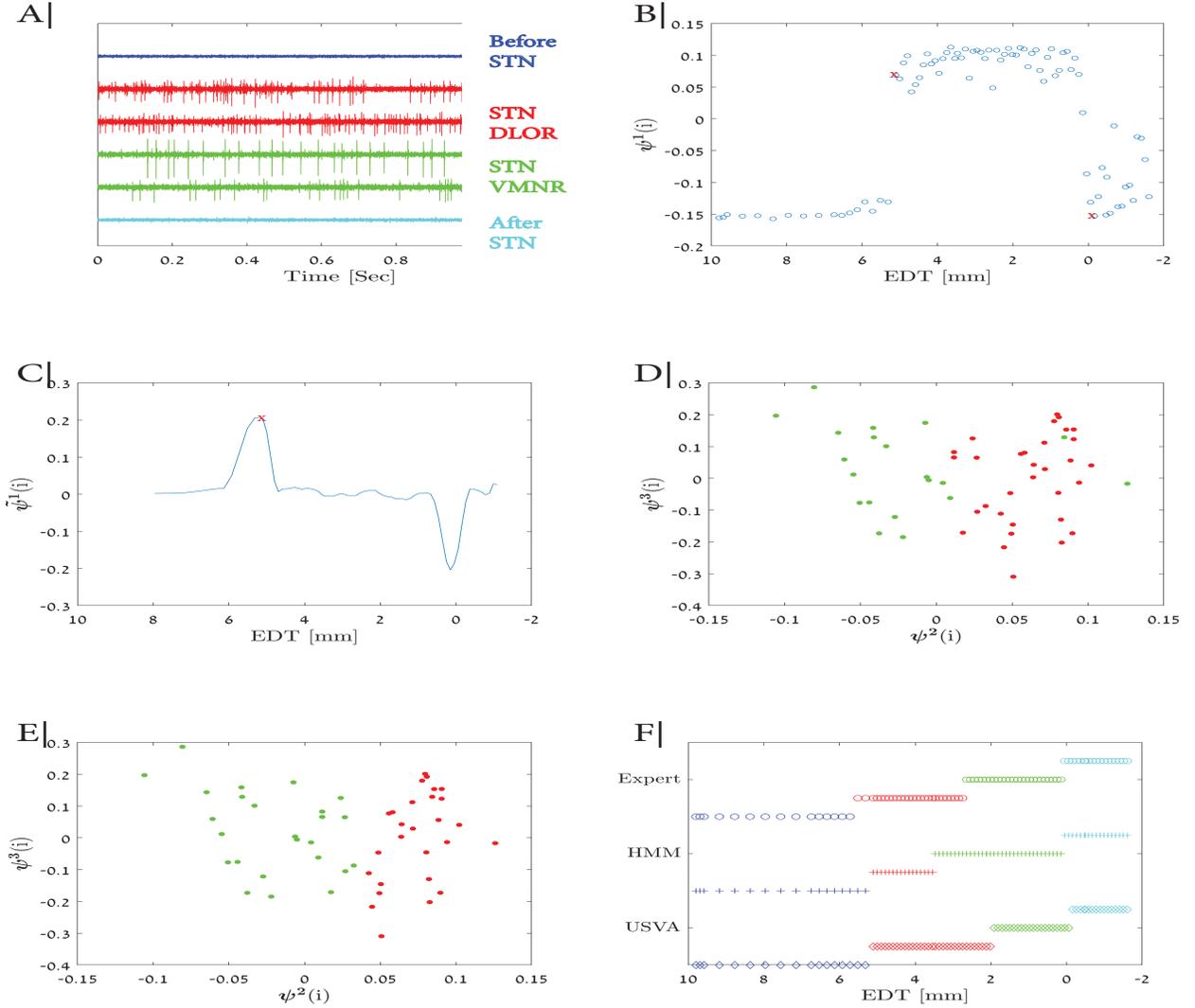}}
\caption{STN border detection obtained by our proposed method -  Unsupervised State Variables approximation (USVA) - applied to a single trajectory. (A) The input data -- time series of measurements of the neuronal activity along the pre-planned trajectory. We show 6 representative raw signal traces out of the 84 signal traces at various depths along the trajectory (a single DBS track) recorded from a Parkinson's disease patient.
The different time series are colored according to the experts' labels: white matter before STN in blue, Dorso-Lateral Oscillatory Region (DLOR) in red, Ventro Medial Non-oscillatory Region (VMNR) in green, and white matter after STN in light blue.
(B) The 1D embedding obtained by the USVA method. The Y-axis displays the approximate state variable value, i.e., the
value of the most dominant eigenvector as a function of the Estimated Distance from Target (EDT). The STN entry and exit locations marked by a human expert are marked by red `x'.
(C) Same as (B) but after pre-processing (smoothing) with a moving average window of size $3$, and computing the difference between the medians of consecutive windows of size $5$.
(D) 2D embedding obtained by the USVA method. The embedded points are colored according to the expert's labels, where red points belong to depths in the DLOR and green points belong to depths in the VMNR. The X-axis indicates the value of the second most dominant eigenvector and the Y-axis indicates the corresponding value of the third most dominant eigenvector.
(E) Same 2D embedding as in (D) but colored according to k-means as suggested in Algorithm \ref{alg:DLOR}. 
(F) A comparison between the detection results obtained by our USVA method (marked with `$\diamond$'), HMM (marked with `+'), and the expert's labels (marked with `o'). Colors are the same as in Figure \ref{fig:DBS}(A).} 
\label{fig:DBS}
\end{figure*}

In this section, we show experimental results on unsupervised detection of target regions during DBS surgery.
We focus on two particular detection tasks: finding the subthalamic nucleus (STN) region and a sub-territory within the STN region, called the dorsolateral oscillatory region (DLOR). 

\subsection{Data and Preprocessing}\label{section:stn_detection}

The dataset was collected at the Hadassah University Medical Central, and it is completely anonymized without any personal identification information. Written informed consent was obtained from all patients and the study was approved by the Institutional Review Board of Hadassah Hospital in accordance with the Helsinki Declaration (reference code: 0168-10-HMO).

The measured signals are time series of neuronal activity at different depths along a pre-planned trajectory recorded by a micro-electrode. The time series at different depths are of varying lengths, depending on the recording time at each depth during the surgery. 

The data was acquired using a Neuro Omega system. The raw data was sampled at 44k Hz by a 16-bit A/D converter (using $\pm 1.25 V$ input range; i.e., $ \sim  2 \mu V$ amplitude resolution). Then, the raw signal was bandpass filtered from 0.075 to 10k Hz, using hardware of 2 and 3 pole Butterworth filters, respectively.

In the context of the problem setting described in Section \ref{section:mehtod}, we refer to each specific depth as a system state, denoted by $i$, for $i={1,..,N}$. Accordingly, let $g_i(\tau_j),  j=1\ldots, T_i$ be the time series of the signal recorded at depth $i$, where $\tau_j$ is the discrete time index and $T_i$ is the length of the signal recorded at depth $i$.

We apply a pre-processing stage to the 1D time series measurements $g_i(\tau_j)$, where we compute the scattering transform \cite{mallat2012group,bruna2015intermittent}. Scattering transform is based on a cascade of wavelet transforms and modulus operators and has been shown to yield an informative representation of time-series measurements.
Let $\boldsymbol{y}_i(t_j)$ denote the resulting scattering transform at time $j= \{1...,M_i\}$ and at depth $i= \{1...,N\} $, where $M_i$ is the number of scattering transform time frames.

The signal acquired at each depth (state of the system) is classified by a human expert into one of four classes: Before STN, STN-DLOR, STN-Ventro Medial Non-oscillatory Region (VMNR), and After STN.

An illustrative example of the data is depicted in Figure \ref{fig:DBS}(A), where we plot the time series measurements from 6 different depths (states), colored according to the experts' labels. 

We observe that signals within the STN region (colored in red and green) have higher variability compared to signals outside the STN region (colored in blue and cyan). Indeed, this variability was used as a feature in previous work, e.g. in \cite{valsky2017s,zaidel2009delimiting,wong2009functional}. We also observe that there is no evident difference between the two classes of signals within the STN region, indicating that DLOR detection is a challenging task.
\subsection{Detection of Sub-Territories of the Subthalamic Nucleus}
\subsubsection{Subthalamic Nucleus (STN) Detection}\label{section:stn_detection}

The proposed algorithm for the detection of the STN region appears in Algorithm \ref{alg:STN}, where we denote by $\text{EDT}(i)$ the $i$th coordinate of the Estimate Distance from Target (EDT) vector designating the specific depth along the pre-planned trajectory. Note that the term EDT is frequently used in this context, and it is known in advance and does not suggest any a-priori knowledge on the target location.

Our empirical examination suggests that the STN location can be determined by the most dominant component, i.e., the entries of eigenvector $\boldsymbol{\psi}^1$ corresponding to the largest eigenvalue. Therefore, the detection of the STN is based only on $\boldsymbol{\psi}^1$ resulting from the application of Algorithm \ref{alg:proposed} with $P=1$ to the pre-processed data.
The detection itself is implemented in steps 3-5.
The main idea is to detect the first sharp transition of values in the entries of $\boldsymbol{\psi}^1$, namely, $\{\boldsymbol{\psi}^1(1), \boldsymbol{\psi}^1(2), \ldots\}$, indicating the entry to the STN region, where the indices of the vector entries represent the depth (state).
In order to alleviate the effect of small perturbations, we smooth the sequence of entries of $\boldsymbol{\psi}^1$ by a moving average with a window of size $3$. Then, we detect the transition by computing the difference between the medians at two consecutive running windows of size $5$ and obtain $\tilde{\boldsymbol{\psi}}^1$. The index at which the maximal difference is obtained is denoted by $i_{\text{en}}$, indicating the depth of the entry point to the STN $\text{EDT}(i_{\text{en}})$.
Once the entry point is determined, the exit point is set as the first point at which $\boldsymbol{\psi}^1(i)$ is smaller than $\boldsymbol{\psi}^1(i_{\text{en}})$. We denote the index of the exit point by $i_{\text{ex}}$ and the corresponding depth by $\text{EDT}(i_{\text{ex}})$.

The result of the application of Algorithm \ref{alg:STN} to the example presented in Figure \ref{fig:DBS}(A) is shown in Figure \ref{fig:DBS}(B)-(C). In Figure \ref{fig:DBS}(B), we plot $\boldsymbol{\psi}^1(i)$ as a function of the depth $\text{EDT}(i)$. In Figure \ref{fig:DBS}(C), we plot $\tilde{\boldsymbol{\psi}}^1(i)$ as a function of the depth $\text{EDT}(i)$.

It is important to note that the eigenvectors are always determined up to a sign. Therefore, in order to eliminate this inherent sign ambiguity, we replace $\boldsymbol{\psi}^1(i)$ by $\text{sign}\left(\delta\right)\cdot\boldsymbol{\psi}^1(i)$, where 
\begin{equation*}
\delta=|\underset{i=2,\ldots,N}{\mathrm{max}}(\tilde{\boldsymbol{\psi}}^1(i))- \tilde{\boldsymbol{\psi}}^1(1)|-|\underset{i=2,\ldots,N}{\mathrm{min}}(\tilde{\boldsymbol{\psi}}^1(i)) - \tilde{\boldsymbol{\psi}}^1(1)|.
\end{equation*}

\begin{algorithm}[t]
	\caption{STN Region Detection}\label{alg:STN}

	\textbf{Input}: ${g_{i}(\tau_j)} \in \mathbb{R}^{T_i}, i={1,..,N}$ -- Time series measurements of neuronal activity at 
    different depths.\\
    $\text{EDT} \in \mathbb{R}^{N}$ -- Vector indicating the Estimated Distance from Target of each depth.\\
	\textbf{Output}: \emph{STN entry point} and \emph{STN exit point}.
	
	\begin{enumerate}
		\item For each time series $g_i(\tau_j)$, compute its Scattering Transform:
		\begin{equation*}
        \boldsymbol{y}_{i}(t_{j})=\Phi(g_i(\tau_j)) \in \mathbb{R}^k,  
		\end{equation*}
		where $j=1,\ldots,M_i$, $i=1,..,N$, and $\Phi$ represents the Scattering transform operator
		\item Compute $\boldsymbol{\psi}^1$ according to Algorithm \ref{alg:proposed}
		\item Compute $\Tilde{\boldsymbol{\psi}}^1$ by applying a moving average with a window of size $3$ samples to $\boldsymbol{\psi}^1$, and then, compute the difference between the medians at two consecutive running windows of size $5$ samples
		\item $i_{\text{en}} =  \underset{i}{\mathrm{argmax}} \quad \Tilde{\boldsymbol{\psi}}^1(i)$
		\item $i_{\text{ex}} =  \underset{i}{\mathrm{argmin}} \quad \{i:\boldsymbol{\psi}^1(i) \leq \boldsymbol{\psi}^1(i_{\text{en}}) \quad \text{and} \quad i >  i_{\text{en}} \}$
		\item Set the \emph{STN entry point} as $\text{EDT}(i_{\text{en}})$, and the \emph{STN exit point} as $\text{EDT}(i_{\text{ex}})$ 

	\end{enumerate}
\end{algorithm}

\subsubsection{Dorsolateral Oscillatory Region (DLOR) Detection}

The proposed algorithm for the detection of the DLOR appears in Algorithm \ref{alg:DLOR}. Since the transitions between the sub-territories of the STN region are subtle and not as distinct as the transition into and out of the STN, we perform two adjustments with respect to Algorithm \ref{alg:STN}. First, we assume that the information on subtle changes in the system's state variables is manifested deeper in the spectrum, namely in eigenvectors corresponding to smaller eigenvalues. Therefore, we use more than one coordinate (eignevectors) to embed the measurements.
Second, we use a small prior on the data -- that the STN region is divided into continuous regions. Accordingly, we modify the diffusion operator as follows:
\begin{align}\label{eq:K_t}
{K}^t = {K} + {K}^s
\end{align}
where ${K}$ is the operator defined in \eqref{eq:K} and is used for the detection of the STN region, and ${K}^s$ is a kernel that enhances temporal proximity, which is given by:
\begin{align}
K^s_{i,l}&=\frac{W^s_{i,l}}{\boldsymbol{w}^s(i)} \quad ,\quad \boldsymbol{w}^s(i)=\sum_{l=1}^{N}W^s_{i,l} \label{eq:K_s_two}
\end{align}
where
\begin{align}
W^s_{i,l}&=\exp\left\{-{\frac{\|\text{EDT}(i)-\text{EDT}(j)\|^2}{\epsilon_s}}\right\},
\end{align}
where $\epsilon_s$ is the kernel scale.

In accordance with the above adjustments, we apply eigenvalue decomposition to ${K}^t$ and represent each depth in the STN region using \emph{two} eigenvectors, $\boldsymbol{\psi}^2$ and $\boldsymbol{\psi}^3$, corresponding to the third and fourth largest eigenvalues. We exclude $\boldsymbol{\psi}^0$ and $\boldsymbol{\psi}^1$ because $\boldsymbol{\psi}^0$ is trivial and $\boldsymbol{\psi}^1$ contains information on the STN boundaries, which is already exploited for the detection of the STN.  
We note that we examined the use of different numbers of eigenvectors, and we choose to represent each depth using two eigenvectors since it led to the best empirical results. Yet, our empirical test suggests that the results are not sensitive to using a different number of eigenvectors.

This representation enables us to find separation in the embedded space. 
Specifically, we apply K-means \cite{hartigan1979algorithm} to the following embedding of each depth within the STN region:
\begin{align}\label{eq:alg_3_rep}
    R(i)=(\boldsymbol{\psi}^2(i),\boldsymbol{\psi}^3(i),\text{EDT}(i)).
\end{align}
We note that since the DLOR is a continuous region we added the third coordinate that encourages temporal continuity of the separation. This third coordinate is appropriately scaled so that it fits the dynamical range of the other two coordinates. We apply K-means with $k=2$, initialized with the entry depth and exit depth of the STN region, and obtain two clusters within the STN. The cluster that includes the entry depth to the STN region is denoted as the DLOR, and the other cluster, which includes the exit depth from the STN region, is denoted as the STN-VMNR.
An example of the 2D embedding of measurements from depths within the STN region is shown in Figure \ref{fig:DBS}(D)-(E). In Figure \ref{fig:DBS}(D), the points are colored according to the labels obtained by a human expert, where red points belong to depths within the DLOR and green points belong to depths labeled as STN-VMNR. The corresponding K-means labels are displayed in Figure \ref{fig:DBS}(E).
Indeed, we see that our unsupervised detection coincides with the expert's labels.
We note that although the embedding in Figure \ref{fig:DBS}(D) may suggest that the classification of the DLOR/STN-VMNR could be based on the sign of $\psi^2$, inspecting other trajectories indicates that it does not apply in general and that the clustering should be based on a combination of $\psi^2$ and $\psi^3$.

Finally, based on Algorithms \ref{alg:STN} and \ref{alg:DLOR}, we cluster the data according to the 4 labels.
A visual comparison between the labels obtained by our unsupervised method and by the supervised HMM algorithm with respect to labels given by an expert to data from a specific example is presented in Figure \ref{fig:DBS}(F). For convenience, our method is denoted by USVA (Unsupervised State Variables Approximation). We see that in the presented example, our unsupervised method is able to detect the STN region with an accuracy that is comparable to the accuracy of the supervised HMM method. In addition, our unsupervised method obtains a superior detection of the DLOR compared to the supervised HMM. We note that these results are with respect to the expert's labels.

\begin{algorithm}[t]
	\caption{DLOR Detection}\label{alg:DLOR}
    
    \textbf{Input}: The affinity matrix ${K} \in \mathbb{R}^{N \times N}$, the STN entry point, the STN exit point\\
    $\text{EDT} \in \mathbb{R}^{N}$ -- the vector indicating the Estimated Distance from Target of each depth.\\
	\textbf{Output}: The \emph{DLOR exit point}.
	
	\begin{enumerate}
		\item Compute a smoothing kernel $K^s$ according to \eqref{eq:K_s_two}
		\item Compute the kernel ${K}^{t}$ according to \eqref{eq:K_t} 
		\item Apply eigenvalue decomposition to ${K}^t$ and obtain its eigenvalues and eigenvectors
		\item Represent each depth in the STN region according to  \eqref{eq:alg_3_rep}, i.e., by $R(i)=(\boldsymbol{\psi}^2(i),\boldsymbol{\psi}^3(i),\text{EDT}(i))$ 
		\item Divide all depth representations $R(i)$ into 2 clusters, DLOR and STN-VMNR, using K-means initialized with the STN entry ($i_{\text{en}}$) and exit ($i_{\text{ex}}$) points
		\item Set $i_{d} =  \underset{i}{\mathrm{argmin}} \{R(i) \in \text{STN} \quad \text{-VMNR} \}$
		\item Set the \emph{DLOR exit point} as $\text{EDT}(i_{d})$ 
	\end{enumerate}
\end{algorithm}

\subsection{Quantitative Detection Results}
\begin{figure*}[!t]
\centerline{\includegraphics[width=151mm,height=90mm]{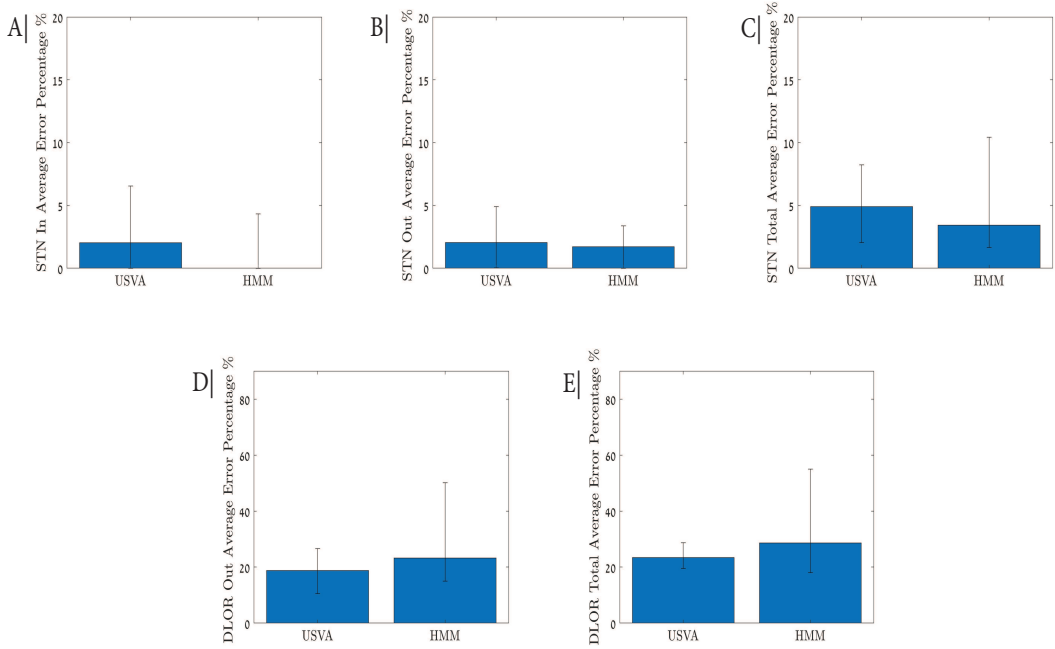}}
\caption{Comparison between our method and the HMM algorithm. The comparison is based on the experts' labels of 25 different trajectories of 16 different patients. The blue thick bars indicate the median error percentage of each task and the black thin bars indicate the interquartile range (IQR). 
(A) The median error in detecting the entry point to the STN region.
(B) The median error in detecting the exit from the STN region.
(C) The median overall error in detecting the STN region.
(D) The median error in detecting the exit from the DLOR.
(E) The median overall error in detecting the DLOR.}
\label{fig:Results}
\end{figure*}

We apply our method (Algorithm \ref{alg:STN} and Algorithm \ref{alg:DLOR}) to $25$ different trajectories recorded from $16$ patients, and we compare the results to the results obtained by the HMM algorithm proposed in \cite{valsky2017s}, which is considered the gold-standard.

Each trajectory consists of three transition points of interest: the STN entry, the DLOR exit, and the STN exit (note that the DLOR entry is the same as the STN entry). In order to quantitatively evaluate the detection, we measure the distance between the transition point marked by the human expert and the detected transition point. For the purpose of normalization, we divide the distance by the size of the respective region.
Consequently, we have a total of six performance measures: STN entry and exit errors (divided by the size of the STN region), DLOR entry and exit errors (divided by the size of the DLOR region), and the overall STN and DLOR errors. Note that the STN entry error and the DLOR entry error differ only by the normalization factor, since the STN entry point coincides with the DLOR entry point. 
The median and interquartile range (IQR) of the five performance measures are reported in Figure \ref{fig:Results}. To complement the experimental study, we also report their mean and standard deviation in percentage in Table \ref{tab:results_table_1} and Table \ref{tab:results_table_2}.

\begin{table}[!ht]
\centering
\caption{The STN borders detection results obtained by the our method(USVA) and the HMM algorithm.
The presented values are the average border detection error (over the $25$ tested trajectories) and the standard deviation with respect to the identification of a human expert. The error is reported in percentage relative to the STN size.}
 \begin{tabular}{||c c c c ||} 
 \hline
 &STN entry&STN exit& STN overall\\  
 \hline
 USVA & $4.77 \pm 6.97$ & $3.95 \pm 5.53$ & $8.72 \pm 11.05$ \\ 
 \hline
 HMM & $4.41 \pm 8.62$ & $2.89 \pm 4.56$ & $7.31 \pm 9.8$\\[0.1ex]
 \hline
\end{tabular}
\label{tab:results_table_1}
\end{table}
\begin{table}[!ht]
\centering
\caption{Same as in Table \ref{tab:results_table_1}, but for the DLOR detection. The error is reported in percentage relative to the DLOR size.}
 \begin{tabular}{|| c c c||} 
 \hline
 & DLOR exit  & DLOR overall \\  
 \hline
 USVA  & $17.89 \pm 13.68$ & $24.49 \pm 17.41$  \\ 
 \hline
 HMM   & $39.16 \pm 34.87$ &$43.86 \pm 37.03$\\[0.1ex]
 \hline 
\end{tabular}
\par
  
\label{tab:results_table_2}
\end{table}
We remark that in our performance evaluation, failures to detect the DLOR exit point is considered to be a $100 \%$ error.
In addition, we note that the DLOR performance measure values are higher than the STN performance measures; this is due to the normalization by the size of the DLOR region, which is smaller than the STN region.

We observe that our method attains results comparable to the gold-standard in the detection of the STN region and outperforms the gold-standard in the detection of the DLOR. Importantly, our method is unsupervised whereas the HMM-based method is supervised, and therefore, could be biased towards the labeling of the specific expert labels, which were used for training.

\section{Discussion} \label{section:discussion}

We presented an agnostic method that can be used during DBS surgery, facilitating physicians in the accurate detection of the STN and the DLOR. After collecting all the MERs during surgery, they can be fed into a system that implements Algorithms \ref{alg:STN} and \ref{alg:DLOR}. The system does not need to be manually tuned and provides recommendations regarding the STN and the DLOR locations. By the nature of our method, these recommendations are not biased towards the labeling of specific experts, and therefore, can advise physicians by helping them to validate their decisions.
It is important to note that our method can recommend a location of the target regions only after all the MERs have been acquired. 

The run-time of our of algorithm is mainly governed by the pre-processing stage of each MER (scattering transform). Therefore, it highly depends on the number of MERs taken from different trajectory depths and the recording time of each MER. After acquiring the MERs from the entire trajectory, the total run-time of our method is approximately 3-5 minutes on a standard personal computer. This run-time can be significantly reduced to a few seconds if we apply the pre-processing stage during the recording of each MER.

We wish to remark that our unsupervised algorithm is able to produce an embedding of measurements from which the STN and the DLOR locations can be accurately identified. However, this embedding does not have an inverse map.
In the context of DBS, if such a map existed, then it could be used to extract the signal properties characterizing the different regions and the transitions between them. In future work, we plan to extend our work in this direction, and attempt to reveal the latent governing variables of the MERs that determine the locations of the STN and the DLOR.

This work can also serve as a stepping-stone for several other research directions. One future research direction is the application of our method to the detection of another brain area, called the Globus Pallidus (GP), which is of interest during a different DBS surgery for treating advanced Parkinson’s disease and dystonia \cite{valsky2020real}. Since the general setup of the GP detection task is very similar to the STN and DLOR detection tasks, we believe that with only mild adjustments our unsupervised method can also be applied there.
Another research direction stems from the fact that we need to compute specially-tailored features based on the measurements without any prior knowledge, and then to estimate their covariance. Finding better estimators for the covariance matrix based on the data, for example using shrinkage \cite{gavish2019optimal,donoho2018optimal}, may significantly improve the method results and enable us to extend the scope to other applications. 
Perhaps the most significant future research direction involves the generalization of the proposed method to multiple sets of measurements from different modalities. 
In the context of manifold learning, multimodal data fusion has attracted much attention recently, e.g. \cite{lederman2018learning,salhov2019multi}. 
We intend to incorporate the proposed variant of the Mahalanobis distance for the purpose of devising unsupervised target and anomaly detection methods for multi-channel and multi-modal data.

\section*{Acknowledgments}
We would like to convey our gratitude to Hadassah University Medical Central for the willingness to share the data with us. We also wish to thank Omer Naor from Alpha Omega for fruitful early stage discussions, and Or Yair, Pavel Lifshits, and Izhak Bucher for their help with the mechanical system experiment.
This research was supported by the Pazy Foundation (grant No. 78-2018) and by the Technion Hiroshi Fujiwara Cyber Security Research Center.
R.T. acknowledges the support of the Schmidt Career Advancement Chair in AI.

\appendices

\section{Proof of Proposition \ref{prop:1}}\label{app1}
The proof of the proposition relies on the following results.
\begin{lemma}
The expected value of the It\^{o} process $\boldsymbol{x}_i(t)$ in \eqref{eq:x_eq} is $E[\boldsymbol{x}_i(t)]=\bar{\boldsymbol{x}}_i$.
\label{lemma:1}
\end{lemma}

\begin{proof}[Proof of Lemma \ref{lemma:1}]
The measurements evolution in time can be represented explicitly by:
\begin{align}
d\boldsymbol{x}_i(t)&=(\bar{\boldsymbol{x}}_i-\boldsymbol{x}_i(t))dt+{\Lambda} d\boldsymbol{w}_i(t)
\end{align} 
The evolution of the $k$th element of $\boldsymbol{x}_i(t)$ is given by 
\begin{align}
dx_i^k&=(\bar{x}_i^k-x_i^k(t))dt+\sigma^k dw_i^k(t) \quad, \quad k = 1,...,d 
\label{eq:xi_eq_lemma}
\end{align}
where $\sigma^k$ is the $k$-th element in the diagonal of ${\Lambda}$. We note that the stochastic process describe in \eqref{eq:xi_eq_lemma} is an Ornstein–Uhlenbeck (OU) process \cite{maller2009ornstein} that satisfies in case of stationarity $\mathbb{E}[x_{i}^k(t)]=\bar{x}_i^k$. 
We note that as an alternative to stationarity, which is assumed throughout, we can assume that the process initial condition is a random variable with expected value $\mathbb{E}[x_{i}^k(0))]=\bar{x}_i^k$. In this case, the expected value of \eqref{eq:xi_eq_lemma} is given by:
\begin{align*}
\mathbb{E}[x_{i}^k(t)]&=\mathbb{E}[x_{i}^k(0)] + \mathbb{E}\int_{0}^{t} (\bar{x}_{i}^k-x_{i}^k(s)) ds \\
	&=\bar{x}_{i}^k+t\bar{x}_{i}^k -\int_{0}^{t} \mathbb{E} [x_{i}^k(s)] ds.
\end{align*}
The ODE above can be recast as:
\begin{align*}
\dot{m}(t) = \bar{x}_{i}^k+t\bar{x}_{i}^k -m(t)
\end{align*}
where $m(t)=\int_{0}^{t} \mathbb{E} [x_{i}^k(s)] ds$. The solution to this ODE is:
\begin{align*}
m(t)=t\bar{x}_i^k,
\end{align*}
and therefore, $\mathbb{E}[x_{i}^k(t)]=\dot{m}(t)=\bar{x}_i^k$.
In vector form, we have $\mathbb{E}[\boldsymbol{x}_{i}(t)]=\bar{\boldsymbol{x}}_i$.
\end{proof}

\begin{lemma}
The covariance matrix of the increments of the It\^{o} process $\boldsymbol{x}_i(t)$ is given by $\mathrm{Cov}[\boldsymbol{x}_i(\delta t+t)|\boldsymbol{x}_i(t)]={\delta t} {\Lambda} ^ 2$.
\label{lemma:2}
\end{lemma}

\begin{proof}[Proof of Lemma \ref{lemma:2}]
By definition, the $(k_1,k_2)$-th element of the covariance matrix is
\begin{align*}
&\mathrm{Cov}[\boldsymbol{x}_i(\delta t+t)|\boldsymbol{x}_i(t)]_{k_1,k_2}\\&=\mathrm{Cov}[x_i^{k_1}(t+\delta t),x_i^{k_2}(t+\delta t)|x_i^{k_1}(t),x_i^{k_2}(t)]\\
&=\delta t (\Lambda ^{k_1 k_2})^2
\end{align*}
where $\Lambda ^{k_1 k_2}$ is the $(k_1,k_2)$-th element of the matrix ${\Lambda}$.
The last equality in the equation above holds since the Brownian motions of the different variables in \eqref{eq:ito_theta}, \eqref{eq:ito_eta}, and \eqref{eq:x_eq} are independent.

In matrix form, we have:
\begin{align*}
\mathrm{Cov}[\boldsymbol{x}_i(\delta t+t)|\boldsymbol{x}_i(t)]&={\delta t} {\Lambda} ^ 2
\end{align*}
\end{proof}

\begin{proof} [Proof of Proposition \ref{prop:1}]
According to Lemma \ref{lemma:1} and Lemma \ref{lemma:2}, the features in \eqref{eq:z_linear} and \eqref{eq:c_overoll} are: $(\boldsymbol{z}_i,{C}_i)=(\bar{\boldsymbol{x}}_i,{\delta t} {\Lambda}_i ^ 2)$.
By using the inverse of the covariance matrix:
\begin{align*}
	{C}_{i}^{-1}=\frac{1}{\delta t} \begin{bmatrix}{I} & 0\\
			0 & \epsilon^2 {I}
	\end{bmatrix}
\end{align*}
and by substituting $(\boldsymbol{z}_i,{C}_i)$ into \eqref{eq:mahalanobis} we obtain that:
\begin{align*}
	d(\boldsymbol{z}_{i},\boldsymbol{z}_{l})&=\frac{1}{2}(\boldsymbol{z}_{i}-\boldsymbol{z}_{l})^{\top}({C}_{i}^{-1}+{C}_{l}^{-1})(\boldsymbol{z}_{i}-\boldsymbol{z}_{l})\\
	&=\frac{1}{\delta t} \sum_{k=1}^{d}e^{k}((\bar{\boldsymbol{x}}_{i})^{k}-(\bar{\boldsymbol{x}}_{l})^{k})^{2}  
\end{align*}
where $e^{k}=1$ if $k\in\{1,..,d_{1}\}$ (denoting the indices of the state variables),
and $e^{k}=\epsilon$ if $k\in\{d_{1}+1,..,d_{1}+d_{2}\}$ (denoting the indices of the noise variables).
Further derivation yields:
\begin{align*}
    d(\boldsymbol{z}_{i},\boldsymbol{z}_{l})&=\frac{1}{\delta t} \sum_{k=1}^{d}e^{k}((\bar{\boldsymbol{x}}_{i})^{k}-(\bar{\boldsymbol{x}}_{l})^{k})^{2}\\
    &=\frac{1}{\delta t}\left[\sum_{k=1}^{d_{1}}((\boldsymbol{\theta}_{i})^{k}-(\boldsymbol{\theta}_{l})^{k})^{2}+\sum_{k=d_{1}+1}^{d}\epsilon((\boldsymbol{\eta}_{i})^{k}-(\boldsymbol{\eta}_{l})^{k})^{2}\right]\\
    &=\frac{1}{\delta t}\left[\|\boldsymbol{\theta}_{i}-\boldsymbol{\theta}_{l}\|^{2}+\epsilon\|\boldsymbol{\eta}_i-\boldsymbol{\eta}_l\|^{2}\right]\\
    &=\frac{1}{\delta t}\left[\|\boldsymbol{\theta}_{i}-\boldsymbol{\theta}_{l}\|^{2}+O(\epsilon)\right]
\end{align*}
\end{proof}

\section{Proof of Proposition \ref{prop:2}}\label{app2}
The proof of the proposition relies on the following results.
\begin{lemma}
The expected value of the sum of $\boldsymbol{y}_i(t)$ and its temporal increment is given by:
\begin{align*}
\mathbb{E}\left[\lim_{\delta t\to 0}\frac{(\boldsymbol{y}_i(t+\delta t)-\boldsymbol{y}_i(t))}{\delta t}+\boldsymbol{y}_i(t)\right] 
&=\bar{\boldsymbol{y}}_i  +  \boldsymbol{\phi}_i\\ &+O(\|\bar{\boldsymbol{x}}_i-\boldsymbol{x}_i(t)\|^2)
\end{align*}
where $\bar{\boldsymbol{y}}_i= \begin{bmatrix} \bar{y}_i^1 \\.. \\\bar{y}_i^d \end{bmatrix}$ and $\boldsymbol{\phi}_i =\mathbb{E} \begin{bmatrix} \frac{1}{2}\sum_{p=1}^{d} (\sigma^{k})^2 f_{pp}^1(\boldsymbol{x}_i(t)) \\.. \\\frac{1}{2}\sum_{k=1}^{d} (\sigma^{k})^2 f_{pp}^d(\boldsymbol{x}_i(t)) \end{bmatrix}$.
\label{lemma:3}
\end{lemma}
\begin{proof} [Proof of Lemma \ref{lemma:3}]
We first note that, the expected value in the equation above is with respect to the joint distribution of the process {$\boldsymbol{x}_i(t)$} and its increments {$\boldsymbol{x}_i(t+\delta t) - \boldsymbol{x}_i(t)$}.
According to It\^{o}'s Lemma, the process $\boldsymbol{y}_i(t)$ is given by:
\begin{align*}
d\boldsymbol{y}_i^k&=[(\bar{\boldsymbol{x}}_i-\boldsymbol{x}_i(t))^\top \nabla f^k(\boldsymbol{x}_i(t))+\frac{1}{2}\sum_{p=1}^{d} (\sigma^{k})^2 f_{pp}^k(\boldsymbol{x}_i(t))]dt \nonumber \\
&+ (\nabla f^k(\boldsymbol{x}_i(t)))^\top {\Lambda} d{W}_i(t) 
\end{align*}
We note that the It\^{o} process $\boldsymbol{x}_i(t)$ is defined as a diffusion process. Therefore, by \cite{lawler2010stochastic} (Chapter 3) and the derivation above, the expected value of the temporal increment of the process $\boldsymbol{y}_i^k(t)=f(\boldsymbol{x}_i(t))$  is given by:
\begin{align}\label{eq:y_i_t_process}
    &\mathbb{E}\left[ \left. \lim_{\delta t\to 0} \frac{\boldsymbol{y}_i^k(\delta t + t)-\boldsymbol{y}_i^k(t)}{\delta t} \right\vert \boldsymbol{x}_i^k(t)\right] \\&= (\bar{\boldsymbol{x}}_i-\boldsymbol{x}_i(t))^\top \nabla f^k(\boldsymbol{x}_i(t)) +\frac{1}{2}\sum_{p=1}^{d} (\sigma^{k})^2 f_{pp}^k(\boldsymbol{x}_i(t)) \nonumber
\end{align}
We use the first order Taylor expansion of $\bar{\boldsymbol{y}}_i^k=f^k(\bar{\boldsymbol{x}}_i)$ around $\boldsymbol{x}_i(t)$:
\begin{align*}
&\bar{\boldsymbol{y}}_i^k= 
\boldsymbol{y}_i(t)+[(\bar{\boldsymbol{x}}_i-\boldsymbol{x}_i(t)]^\top \nabla f^k(\boldsymbol{x}_i(t))+O(\|\bar{\boldsymbol{x}}_i-\boldsymbol{x}_i(t)\|^2)
\end{align*}
and substitute it into \eqref{eq:y_i_t_process} and get:
\begin{align*}
    &\mathbb{E}\left[\left. \lim_{\delta t\to 0} \frac{\boldsymbol{y}_i^k(\delta t + t)-\boldsymbol{y}_i^k(t)}{\delta t} \right\vert \boldsymbol{x}_i^k(t)\right] \\&=\bar{\boldsymbol{y}}_i^k - \boldsymbol{y}_i^k(t) +\frac{1}{2}\sum_{p=1}^{d} (\sigma^{k})^2 f_{pp}^k(\boldsymbol{x}_i(t)) \nonumber
\end{align*}
After rearranging the the equation we get:
\begin{align*}
    &\mathbb{E}\left[\left.\lim_{\delta t\to 0} \frac{\boldsymbol{y}_i^k(\delta t + t)-\boldsymbol{y}_i^k(t)}{\delta t} +\boldsymbol{y}_i^k(t)  \right\vert \boldsymbol{x}_i^k(t)\right] \\&=\bar{\boldsymbol{y}}_i^k +\frac{1}{2}\sum_{p=1}^{d} (\sigma^{k})^2 f_{pp}^k(\boldsymbol{x}_i(t)) \nonumber
\end{align*}
We now use the law of total expectation, and take the expected value with respect to the process $\boldsymbol{x}_i^k(t)$ of both sides of the equation:
\begin{align*}
    &\mathbb{E}\left[\lim_{\delta t\to 0} \frac{\boldsymbol{y}_i^k(\delta t + t)-\boldsymbol{y}_i^k(t)}{\delta t} +\boldsymbol{y}_i^k(t)\right]\\
    &=\mathbb{E}\left[\left. \mathbb{E}\left[\lim_{\delta t\to 0} \frac{\boldsymbol{y}_i^k(\delta t + t)-\boldsymbol{y}_i^k(t)}{\delta t} +\boldsymbol{y}_i^k(t)  \right\vert \boldsymbol{x}_i^k(t)\right]\right] \\
    &=\mathbb{E}\left[\bar{\boldsymbol{y}}_i^k +\frac{1}{2}\sum_{p=1}^{d} (\sigma^{k})^2 f_{pp}^k(\boldsymbol{x}_i(t))\right]\\
    &=\bar{\boldsymbol{y}}_i^k +\mathbb{E}\left[\frac{1}{2}\sum_{p=1}^{d} (\sigma^{k})^2 f_{pp}^k(\boldsymbol{x}_i(t))\right] \nonumber
\end{align*}
The proof is concluded by rewriting the above in a matrix form:
\begin{align*}
\mathbb{E}\left[\lim_{\delta t\to 0}\frac{\boldsymbol{y}_i(\delta t+ t)-\boldsymbol{y}_i(t)}{\delta t} + \boldsymbol{y}_i(t)\right]&= \bar{\boldsymbol{y}}_i+\boldsymbol{\phi}_i\\\\
 &+ O(\|\bar{\boldsymbol{x}}_i-\boldsymbol{x}_i(t)\|^2)
\end{align*}
whereas we noted before the last expected value is with respect to  the joint distribution of the process and its increments.
\end{proof}

\begin{lemma}
The covariance ${C}_i=\mathrm{Cov}[\boldsymbol{y}_i(t+\delta t)| \boldsymbol{y}_i(t)]$ is given by:
\begin{align*}
{C}_i=\mathrm{Cov}[\boldsymbol{y}_i(t+\delta t)| \boldsymbol{y}_i(t)]={J}{\Lambda} ^2 {J}^\top + O(\delta t)   
\end{align*}
where ${J}$ is the $s \times d $ Jacobian matrix of the function $f$ and $\bar{\boldsymbol{y}}_i=f(\bar{\boldsymbol{x}}_i)$.
\label{lemma:4}
\end{lemma}
\begin{proof}[Proof of Lemma \ref{lemma:4}]
The full proof of Lemma \ref{lemma:4} appears in \cite{Coifman_Singer:2008} and \cite{dsilva2016data}. In \cite{dsilva2016data}, the authors show that the covariance at a specific point in time $t$ of the process $\boldsymbol{y}_i(t)$ can be approximated by short trajectories of the process initialized at that point. The second equality is shown in \cite{Coifman_Singer:2008}.
\end{proof}
\begin{proof} [Proof of Proposition \ref{prop:2}]
 According to Lemma \ref{lemma:3} and Lemma \ref{lemma:4}, the modified features \eqref{eq:z_i_nonlinear} are: $(\boldsymbol{z}_i,{C}_i)=(\bar{\boldsymbol{y}}_i  +  \boldsymbol{\phi}_i,{J}{\Lambda} ^2 {J}^\top)$.
The proposed distance with the modified features is given by: 
\begin{align}\label{eq:prop2_proof}
d(\boldsymbol{z}_i,\boldsymbol{z}_l)&=\frac{1}{2}(\boldsymbol{z}_i-\boldsymbol{z}_l)^\top({C}_{i}^{-1}+{C}_{l}^{-1} )(\boldsymbol{z}_i-\boldsymbol{z}_l)\\
&=\frac{1}{2}(\bar{\boldsymbol{y}}_i-\bar{\boldsymbol{y}}_l)^\top(({J}{\Lambda}^2{J}^\top )^{-1}+({J}{\Lambda}^2{J}^\top )^{-1})(\bar{\boldsymbol{y}}_i-\bar{\boldsymbol{y}}_l) \nonumber\\\
&+\frac{1}{2}(\boldsymbol{\phi}_i-\boldsymbol{\phi}_l)^\top(({J}{\Lambda}^2{J}^\top )^{-1}+({J}{\Lambda}^2{J}^\top )^{-1} )(\boldsymbol{\phi}_i-\boldsymbol{\phi}_l) \nonumber
\end{align}
According to \cite{Coifman_Singer:2008}, the first term on the right hand size of \eqref{eq:prop2_proof} is equal to:
\begin{align*}
&\frac{1}{2}(\bar{\boldsymbol{y}}_i-\bar{\boldsymbol{y}}_l)^{\top}(({J}{\Lambda}^2{J}^{\top} )^{-1}+({J}{\Lambda}^2{J}^{\top} )^{-1})(\bar{\boldsymbol{y}}_i-\bar{\boldsymbol{y}}_l)=\\&\|\bar{\boldsymbol{x}}_i-\bar{\boldsymbol{x}}_l\|_M +O(\|\bar{\boldsymbol{y}}_i-\bar{\boldsymbol{y}}_l\|^4) 
\end{align*}
where $\bar{\boldsymbol{x}}_i=\begin{bmatrix} \bar{\boldsymbol{\theta}}_{i} \\ \bar{\boldsymbol{\eta}}_i \end{bmatrix}$, and for convenience, we denote the norm associated with the modified Mahalanobis distance in direct access case by $\|\bar{\boldsymbol{x}}_i-\bar{\boldsymbol{x}}_l\|_M = \frac{1}{2}(\bar{\boldsymbol{x}}_{i}-\bar{\boldsymbol{x}}_{l})^{\top} {\Lambda}^{-1}(\bar{\boldsymbol{x}}_{i}-\bar{\boldsymbol{x}}_{l})$.
By Assumption \ref{assumption:1}, the second term on the right hand side of \eqref{eq:prop2_proof} is of order $O(\epsilon)$, where $\epsilon \ll 1$.

Combining the above, using Proposition \ref{prop:1}, we have:
\begin{align*}
d(\boldsymbol{z}_i,\boldsymbol{z}_l)&=\|\bar{\boldsymbol{x}}_i-\bar{\boldsymbol{x}}_l\|_M +O(\|\bar{\boldsymbol{y}}_i-\bar{\boldsymbol{y}}_l\|^4) + O(\epsilon)\\
&=\|\bar{\boldsymbol{\theta}}_i-\bar{\boldsymbol{\theta}}_l\|^2  +O(\|\bar{\boldsymbol{y}}_i-\bar{\boldsymbol{y}}_l\|^4) + O(\epsilon)
\end{align*}
\end{proof}

\section{Linear Case Analysis}
We consider a special case where the measurement function is a linear transformation $\boldsymbol{y}=f(\boldsymbol{x})=Ax$ where $A \in \mathbb{R}^{d \times s}$. In \cite{dsilva2016data}, a similar scenario was studied, and here, we primarily review the main steps of their analysis with minor adjustments to our problem setup.
In case $f$ is some linear function, the covariance of the increments of the process  $\boldsymbol{y}_i(t)$ is $C_i=\mathrm{Cov} [\boldsymbol{y}_i(t+\delta t)|\boldsymbol{y}_i(t)] =A \Lambda A^T $, where $\Lambda=\mathrm{Cov}[\boldsymbol{x}_i(t+\delta t)|\boldsymbol{x}_i(t)] $ is the covariance of the increment of $\boldsymbol{x}_i(t)$. 
Let $A=U \Sigma V^T$  be the singular value decomposition (SVD) of A.
The pseudo-inverse of the covariance matrix $C_i$ is:
\begin{align}\label{eq:lin_C}
C_i^{\dagger}=U \Sigma^{-1} V^T \Lambda^{-1} V \Sigma^{-1} U^T    
\end{align}
In addition, in this case, the proposed feature  $\boldsymbol{z}_i=\mathbb{E}[\boldsymbol{y}_i(t)]$ satisfies:
\begin{align}\label{eq:lin_z}
    \boldsymbol{z}_i=\mathbb{E}[\boldsymbol{y}_i(t)] = \bar{\boldsymbol{y}}_i=A\bar{\boldsymbol{x}_i}
\end{align}
By \eqref{eq:lin_C} and \eqref{eq:lin_z}, the proposed distance can be recast as:
\begin{align*}
&d(\boldsymbol{z}_{i},\boldsymbol{z}_{l})=\frac{1}{2}(\bar{\boldsymbol{y}}_i-\bar{\boldsymbol{y}}_l)^T (C_i^{\dagger} +C_l^{\dagger})(\bar{\boldsymbol{y}}_i-\bar{\boldsymbol{y}}_l)\\
&=(\bar{\boldsymbol{x}}_i-\bar{\boldsymbol{x}}_l)^T A^T C_i^{\dagger} A (\bar{\boldsymbol{x}}_i-\bar{\boldsymbol{x}}_l)\\
&=(\bar{\boldsymbol{x}}_i-\bar{\boldsymbol{x}}_l)^T \Lambda^{-1} (\bar{\boldsymbol{x}}_i-\bar{\boldsymbol{x}}_l)\\
&=\|\bar{\boldsymbol{x}}_i-\bar{\boldsymbol{x}}_l\|_M\\
&=\frac{1}{\delta t}\left[\|\boldsymbol{\theta}_{i}-\boldsymbol{\theta}_{l}\|^{2}+O(\epsilon)\right]
\end{align*} 
where the last equality is due to Proposition \ref{prop:1}.

\bibliographystyle{unsrt}  
\bibliography{templatePRIME}

\begin{thebibliography}{10}

\bibitem{limousin1998electrical}
Patricia Limousin, Paul Krack, Pierre Pollak, AbdelHamid Benazzouz, Claire
  Ardouin, Dominique Hoffmann, and Alim-Louis Benabid.
\newblock Electrical stimulation of the subthalamic nucleus in advanced
  parkinson's disease.
\newblock {\em New England Journal of Medicine}, 339(16):1105--1111, 1998.

\bibitem{benabid1994acute}
AL~Benabid, P~Pollak, Ch~Gross, D~Hoffmann, A~Benazzouz, DM~Gao, A~Laurent,
  M~Gentil, and J~Perret.
\newblock Acute and long-term effects of subthalamic nucleus stimulation in
  parkinson's disease.
\newblock {\em Stereotactic and functional neurosurgery}, 62(1-4):76--84, 1994.

\bibitem{hariz2002complications}
Marwan~I Hariz.
\newblock Complications of deep brain stimulation surgery.
\newblock {\em Movement disorders: official journal of the Movement Disorder
  Society}, 17(S3):S162--S166, 2002.

\bibitem{nickl2019rescuing}
Robert~C Nickl, Martin~M Reich, Nicol{\'o}~Gabriele Pozzi, Patrick Fricke,
  Florian Lange, Jonas Roothans, Jens Volkmann, and Cordula Matthies.
\newblock Rescuing suboptimal outcomes of subthalamic deep brain stimulation in
  parkinson disease by surgical lead revision.
\newblock {\em Neurosurgery}, 85(2):E314--E321, 2019.

\bibitem{moro2002impact}
E~Moro, RJA Esselink, J~Xie, M~Hommel, AL~Benabid, and P~Pollak.
\newblock The impact on parkinson’s disease of electrical parameter settings
  in stn stimulation.
\newblock {\em Neurology}, 59(5):706--713, 2002.

\bibitem{witt2012factors}
Karsten Witt, Christine Daniels, and Jens Volkmann.
\newblock Factors associated with neuropsychiatric side effects after stn-dbs
  in parkinson's disease.
\newblock {\em Parkinsonism \& related disorders}, 18:S168--S170, 2012.

\bibitem{moran2006real}
Anan Moran, Izhar Bar-Gad, Hagai Bergman, and Zvi Israel.
\newblock Real-time refinement of subthalamic nucleus targeting using bayesian
  decision-making on the root mean square measure.
\newblock {\em Movement disorders: official journal of the Movement Disorder
  Society}, 21(9):1425--1431, 2006.

\bibitem{weinberger2006beta}
Moran Weinberger, Neil Mahant, William~D Hutchison, Andres~M Lozano, Elena
  Moro, Mojgan Hodaie, Anthony~E Lang, and Jonathan~O Dostrovsky.
\newblock Beta oscillatory activity in the subthalamic nucleus and its relation
  to dopaminergic response in parkinson's disease.
\newblock {\em Journal of neurophysiology}, 96(6):3248--3256, 2006.

\bibitem{zaidel2010subthalamic}
Adam Zaidel, Alexander Spivak, Benjamin Grieb, Hagai Bergman, and Zvi Israel.
\newblock Subthalamic span of $\beta$ oscillations predicts deep brain
  stimulation efficacy for patients with parkinson’s disease.
\newblock {\em Brain}, 133(7):2007--2021, 2010.

\bibitem{shamir2012microelectrode}
Reuben~R Shamir, Adam Zaidel, Leo Joskowicz, Hagai Bergman, and Zvi Israel.
\newblock Microelectrode recording duration and spatial density constraints for
  automatic targeting of the subthalamic nucleus.
\newblock {\em Stereotactic and functional neurosurgery}, 90(5):325--334, 2012.

\bibitem{novak2007detection}
Peter Novak, Slawomir Daniluk, Samuel~A Ellias, and Jules~M Nazzaro.
\newblock Detection of the subthalamic nucleus in microelectrographic
  recordings in parkinson disease using the high-frequency (> 500 hz) neuronal
  background.
\newblock {\em Journal of neurosurgery}, 106(1):175--179, 2007.

\bibitem{telkes2016prediction}
Ilknur Telkes, Joohi Jimenez-Shahed, Ashwin Viswanathan, Aviva Abosch, and
  Nuri~F Ince.
\newblock Prediction of stn-dbs electrode implantation track in parkinson's
  disease by using local field potentials.
\newblock {\em Frontiers in Neuroscience}, 10:198, 2016.

\bibitem{valsky2017s}
Dan Valsky, Odeya Marmor-Levin, Marc Deffains, Renana Eitan, Kim~T Blackwell,
  Hagai Bergman, and Zvi Israel.
\newblock S top! border ahead: A utomatic detection of subthalamic exit during
  deep brain stimulation surgery.
\newblock {\em Movement Disorders}, 32(1):70--79, 2017.

\bibitem{wong2009functional}
Stephen Wong, GH~Baltuch, JL~Jaggi, and SF~Danish.
\newblock Functional localization and visualization of the subthalamic nucleus
  from microelectrode recordings acquired during dbs surgery with unsupervised
  machine learning.
\newblock {\em Journal of neural engineering}, 6(2):026006, 2009.

\bibitem{zaidel2009delimiting}
Adam Zaidel, Alexander Spivak, Lavi Shpigelman, Hagai Bergman, and Zvi Israel.
\newblock Delimiting subterritories of the human subthalamic nucleus by means
  of microelectrode recordings and a hidden markov model.
\newblock {\em Movement disorders}, 24(12):1785--1793, 2009.

\bibitem{coifman2006diffusion}
Ronald~R Coifman and St{\'e}phane Lafon.
\newblock Diffusion maps.
\newblock {\em Applied and computational harmonic analysis}, 21(1):5--30, 2006.

\bibitem{mattingly2010convergence}
Jonathan~C Mattingly, Andrew~M Stuart, and Michael~V Tretyakov.
\newblock Convergence of numerical time-averaging and stationary measures via
  poisson equations.
\newblock {\em SIAM Journal on Numerical Analysis}, 48(2):552--577, 2010.

\bibitem{meyn2012markov}
Sean~P Meyn and Richard~L Tweedie.
\newblock {\em Markov chains and stochastic stability}.
\newblock Springer Science \& Business Media, 2012.

\bibitem{dsilva2016data}
Carmeline~J Dsilva, Ronen Talmon, C~William Gear, Ronald~R Coifman, and
  Ioannis~G Kevrekidis.
\newblock Data-driven reduction for a class of multiscale fast-slow stochastic
  dynamical systems.
\newblock {\em SIAM Journal on Applied Dynamical Systems}, 15(3):1327--1351,
  2016.

\bibitem{Coifman_Singer:2008}
Amit Singer and Ronald~R. Coifman.
\newblock Non-linear independent component analysis with diffusion maps.
\newblock 25(2):226 -- 239, 2008.

\bibitem{kushnir2012anisotropic}
Dan Kushnir, Ali Haddad, and Ronald~R Coifman.
\newblock Anisotropic diffusion on sub-manifolds with application to earth
  structure classification.
\newblock {\em Applied and Computational Harmonic Analysis}, 32(2):280--294,
  2012.

\bibitem{Tenenbaum2000}
Joshua~B. Tenenbaum, Vin de~Silva, and John~C. Langford.
\newblock A global geometric framework for nonlinear dimensionality reduction.
\newblock {\em Science}, 260:2319--2323, 2000.

\bibitem{Roweis2000}
Sam~T. Roweis and Lawrence~K. Saul.
\newblock Nonlinear dimensionality reduction by locally linear embedding.
\newblock {\em Science}, 260:2323--2326, 2000.

\bibitem{Donoho2003}
David~L. Donoho and Carrie Grimes.
\newblock Hessian eigenmaps: New locally linear embedding techniques for
  high-dimensional data.
\newblock {\em Proc. Nat. Acad. Sci.}, 100:5591--5596, 2003.

\bibitem{Belkin_Niyogi_2003}
Mikhail Belkin and Parth Niyogi.
\newblock {Laplacian Eigenmaps for Dimensionality Reduction and Data
  Representation}.
\newblock {\em Neural. Comput.}, 15(6):1373--1396, June 2003.

\bibitem{Coifman_Lafon2006}
Ronald.~R. Coifman and St{\'e}phan Lafon.
\newblock Diffusion maps.
\newblock 21(1):5--30, 2006.

\bibitem{talmon2015manifold}
Ronen Talmon, St{\'e}phane Mallat, Hitten Zaveri, and Ronald~R Coifman.
\newblock Manifold learning for latent variable inference in dynamical systems.
\newblock {\em IEEE Transactions on Signal Processing}, 63(15):3843--3856,
  2015.

\bibitem{yair2017reconstruction}
Or~Yair, Ronen Talmon, Ronald~R Coifman, and Ioannis~G Kevrekidis.
\newblock Reconstruction of normal forms by learning informed observation
  geometries from data.
\newblock {\em Proceedings of the National Academy of Sciences},
  114(38):E7865--E7874, 2017.

\bibitem{zhu2018image}
Bo~Zhu, Jeremiah~Z Liu, Stephen~F Cauley, Bruce~R Rosen, and Matthew~S Rosen.
\newblock Image reconstruction by domain-transform manifold learning.
\newblock {\em Nature}, 555(7697):487--492, 2018.

\bibitem{singer2009diffusion}
Amit Singer, Yoel Shkolnisky, and Boaz Nadler.
\newblock Diffusion interpretation of nonlocal neighborhood filters for signal
  denoising.
\newblock {\em SIAM Journal on Imaging Sciences}, 2(1):118--139, 2009.

\bibitem{ibanez2018manifold}
Rub{\'e}n Ibanez, Emmanuelle Abisset-Chavanne, Jose~Vicente Aguado, David
  Gonzalez, Elias Cueto, and Francisco Chinesta.
\newblock A manifold learning approach to data-driven computational elasticity
  and inelasticity.
\newblock {\em Archives of Computational Methods in Engineering}, 25(1):47--57,
  2018.

\bibitem{shnitzer2019recovering}
Tal Shnitzer, Mirela Ben-Chen, Leonidas Guibas, Ronen Talmon, and Hau-Tieng Wu.
\newblock Recovering hidden components in multimodal data with composite
  diffusion operators.
\newblock {\em SIAM Journal on Mathematics of Data Science}, 1(3):588--616,
  2019.

\bibitem{wu2014assess}
Hau-tieng Wu, Ronen Talmon, and Yu-Lun Lo.
\newblock Assess sleep stage by modern signal processing techniques.
\newblock {\em IEEE Transactions on Biomedical Engineering}, 62(4):1159--1168,
  2014.

\bibitem{talmon2013empirical}
Ronen Talmon and Ronald~R Coifman.
\newblock Empirical intrinsic geometry for nonlinear modeling and time series
  filtering.
\newblock {\em Proceedings of the National Academy of Sciences},
  110(31):12535--12540, 2013.

\bibitem{mallat2012group}
St{\'e}phane Mallat.
\newblock Group invariant scattering.
\newblock {\em Communications on Pure and Applied Mathematics},
  65(10):1331--1398, 2012.

\bibitem{bruna2015intermittent}
Joan Bruna, St{\'e}phane Mallat, Emmanuel Bacry, Jean-Fran{\c{c}}ois Muzy,
  et~al.
\newblock Intermittent process analysis with scattering moments.
\newblock {\em The Annals of Statistics}, 43(1):323--351, 2015.

\bibitem{hartigan1979algorithm}
John~A Hartigan and Manchek~A Wong.
\newblock Algorithm as 136: A k-means clustering algorithm.
\newblock {\em Journal of the Royal Statistical Society. Series C (Applied
  Statistics)}, 28(1):100--108, 1979.

\bibitem{valsky2020real}
Dan Valsky, Kim~T Blackwell, Idit Tamir, Renana Eitan, Hagai Bergman, and Zvi
  Israel.
\newblock Real-time machine learning classification of pallidal borders during
  deep brain stimulation surgery.
\newblock {\em Journal of Neural Engineering}, 17(1):016021, 2020.

\bibitem{gavish2019optimal}
Matan Gavish, Ronen Talmon, Pei-Chun Su, and Hau-Tieng Wu.
\newblock Optimal recovery of mahalanobis distance in high dimension.
\newblock {\em arXiv preprint arXiv:1904.09204}, 2019.

\bibitem{donoho2018optimal}
David~L Donoho, Matan Gavish, and Iain~M Johnstone.
\newblock Optimal shrinkage of eigenvalues in the spiked covariance model.
\newblock {\em Annals of statistics}, 46(4):1742, 2018.

\bibitem{lederman2018learning}
Roy~R Lederman and Ronen Talmon.
\newblock Learning the geometry of common latent variables using
  alternating-diffusion.
\newblock {\em Applied and Computational Harmonic Analysis}, 44(3):509--536,
  2018.

\bibitem{salhov2019multi}
Moshe Salhov, Ofir Lindenbaum, Yariv Aizenbud, Avi Silberschatz, Yoel
  Shkolnisky, and Amir Averbuch.
\newblock Multi-view kernel consensus for data analysis.
\newblock {\em Applied and Computational Harmonic Analysis}, 49(1):208--228,
  2020.

\bibitem{maller2009ornstein}
Ross~A Maller, Gernot M{\"u}ller, and Alex Szimayer.
\newblock Ornstein--uhlenbeck processes and extensions.
\newblock {\em Handbook of financial time series}, pages 421--437, 2009.

\bibitem{lawler2010stochastic}
Gregory~F Lawler.
\newblock Stochastic calculus: An introduction with applications.
\newblock {\em American Mathematical Society}, 2010.

\end{thebibliography}

\end{document}